\documentclass[10pt]{article}
 
\usepackage[a4paper,left=17mm,top=25mm,right=17mm,bottom=25mm, marginparsep=.1in]{geometry}

\usepackage{authblk}
\usepackage[numbers, sort&compress]{natbib}
\usepackage[margin=30pt,font=small,labelfont=bf]{caption}
\usepackage{amsmath,amsopn,amsthm,amssymb}

\usepackage{tikz}
\usetikzlibrary{positioning,arrows.meta}

\usepackage{graphicx}
\usepackage{mathptmx} 
\usepackage[mathscr]{euscript}
\usepackage{mathtools}
\usepackage{float}
\usepackage{xspace}
\usepackage{scalerel}
\usepackage{marginnote}
\usepackage{enumitem}
\setlist[enumerate,1]{label={(\arabic*)}}

\usepackage[
 	colorlinks=true,
	urlcolor=black,
	linkcolor=blue
]{hyperref}

\newtheorem{theorem}{Theorem}[section]
\newtheorem{proposition}[theorem]{Proposition}

\newtheorem{lemma}[theorem]{Lemma} 
\newtheorem{corollary}[theorem]{Corollary}

\newtheorem{definition}[theorem]{Definition}
\newtheorem{observation}[theorem]{Observation}

\newtheorem*{remark}{Remark}

\newcommand{\lca}{\ensuremath{\operatorname{lca}}}
\newcommand{\Hasse}{\mathscr{H}}
\newcommand{\srel}{\blacktriangleleft}
\newcommand{\rel}{\trianglelefteq}
\newcommand{\axiom}[1]{\textnormal{\textbf{(#1)}}}
\newcommand{\pairs}{\mathcal{P}_2}
\newcommand{\support}{\ensuremath{\operatorname{supp}}}
\DeclareMathOperator{\parent}{par}
\DeclareMathOperator{\indeg}{indeg}
\DeclareMathOperator{\outdeg}{outdeg}
\DeclareMathOperator{\CC}{\mathtt{C}}

\newcommand{\cG}{\mathscr{G}}

\newcommand{\lcaV}{\mathfrak{lca}_2}
\newcommand{\notlcaV}{\overline{\mathfrak{lca}}_2}
\newcommand{\reg}{\textsc{Reg}_2}

    {\begin{description}[leftmargin = 0.2cm, labelsep = 0.2cm]}
    {\end{description}}

\providecommand{\keywords}[1]{\textbf{\textit{Keywords: }} #1}

\title{Encoding Phylogenetic Networks with Least Common Ancestor Constraints}
\small

\author[1,*]{Marc Hellmuth} 

\author[1]{Anna Lindeberg} 

\author[2]{Vincent Moulton} 

\affil[1]{Department of Mathematics, Faculty of Science,
  Stockholm University, SE-10691 Stockholm, Sweden} 

\affil[2]{School of Computing Sciences, 
University of East Anglia, Norwich, NR4 7TJ, United Kingdom}

\affil[*]{corresponding author}

\date{\ }

\begin{document}
\sloppy

\maketitle

\abstract{ 
Encoding phylogenetic networks by suitable substructures is a central problem in phylogenetic
combinatorics. In this paper, we study a type of encoding based on least common ancestor (LCA)
constraints. For a directed acyclic graph $G$ with leaf set $X$, we consider the relation $\rel_G$
on pairs of leaves, where $(ab,xy)\in\rel_G$ records that the least common ancestors $\lca_G(ab)$
and $\lca_G(xy)$ for the leaves $a,b,x,y$ in $X$ are well-defined and that $\lca_G(ab)$ is a descendant of $\lca_G(xy)$. Thus,
$\rel_G$ captures the ancestor order among all well-defined pairwise LCAs of $G$.

We first identify precisely which part of a directed acyclic graph (DAG) or network $G$ 
is determined by the relation $\rel_G$. To this end, we
compare the so-called canonical DAG $\cG_{\rel_G}$ that is constructed from $\rel_G$ with the 2-regularization
$\reg(G)$ that is obtained by removing vertices from $G$ that are not LCAs of one or two leaves and then deleting
shortcut edges. More specifically, we prove that $\cG_{\rel_G}$ and $\reg(G)$ are isomorphic. Hence, the obstruction to
encoding a DAG by $\rel_G$ is exactly the information lost under 2-regularization.

This yields a general reconstruction principle, which we apply to several natural classes of
phylogenetic networks. In particular, we show, among others, that shortcut-free 2-lca-relevant DAGs, phylogenetic
trees, regular level-1 networks, regular networks with binary clustering systems, regular networks
whose clustering systems are closed weak hierarchies, strong-phylogenetic normal networks, separated
phylogenetic normal networks, and binary normal networks are all encoded by $\rel$. 

We also introduce the sparse triple-like restriction $\rel_G^3$ of $\rel_G$, consisting only of comparisons of the form $(ab,ac)$ with
$a,b,c$ in $X$ pairwise distinct. For DAGs with the 2-lca-property, we show that $\rel_G^3$, together with
the leaf set, determines $\rel_G$ after a natural closure operation. Consequently, several of the
above classes can be reconstructed, up to isomorphism, from $\rel_G^3$ in polynomial time.
}

\smallskip
\noindent
\keywords{DAG; least common ancestors; encoding; closure; regular networks; normal networks; level-1 networks; 2-regularization}


\section{Introduction}

Encoding phylogenetic networks by the substructures they display is a central theme in phylogenetic
combinatorics \cite{dress2011basic}. 
Given a class $\mathcal C$ of networks and an associated collection $\mathscr S_N$ of
substructures of a network $N$, one asks if any two networks $N,N'$ in this
class are isomorphic if and only if they have the same substructures, i.e. whether or not the following holds:
\[
    \mathscr S_N= \mathscr S_{N'} \quad\Longleftrightarrow\quad N\simeq N'
    \qquad\text{for all } N,N'\in\mathcal C .
\]
If this holds, then $\mathcal C$ is said to be \emph{encoded by} $\mathscr S$. Classical examples
include the encoding of phylogenetic trees by their displayed rooted triples
\cite{Aho:81,sem-ste-03a}, the encoding of level-1 networks by trinets \cite{HM:2013}, 
the encoding of several classes of networks by their clustering systems \cite{Hellmuth2023},
the
encoding of regular networks by their underlying trees \cite{willson2010regular}
and the encoding of binary normal networks by  their underlying caterpillars 
or trees \cite{linz2020caterpillars,BSS:26}. 
Other types of
substructures have also been considered, including subnetworks obtained by deleting reticulate
edges, and more general restrictions to bounded-level or tree-child
network classes \cite{Gambette2012,vIM:14,Murakami:19,Semple2021,francis:25b, 
Pardi:2015,Pardi:2019corr}. These
results show that the type of substructure needed to recover a network depends strongly on the class
of networks under consideration.

In this paper, we consider when a different type of information encodes a network or not, 
namely constraints induced by least
common ancestors. To be more precise, for vertices $v$ and $w$ of a DAG or network $G$, we write
$w\preceq_G v$ if $w$ is a descendant of $v$ in $G$. A least common ancestor, or LCA, of two leaves
$x$ and $y$ is a $\preceq_G$-minimal common ancestor of $x$ and $y$. If this vertex is unique, we
denote it by $\lca_G(xy)$ and say that $\lca_G(xy)$ is well-defined. For a DAG $G$ with leaf set
$X$, we associate a relation $\mathscr S_G=\rel_G$ by putting \[ \rel_G \coloneqq \{(ab,xy) \mid
\lca_G(ab),\lca_G(xy) \text{ are well-defined and } \lca_G(ab)\preceq_G \lca_G(xy)\}. \] 
Thus, $\rel_G$ records the relative ancestor order of all well-defined LCAs of pairs of leaves. We ask
when this information encodes the underlying DAG or phylogenetic network.

At first sight, one might hope that $\rel_G$ encodes $G$. This is not true in general. For example,
adding shortcuts, redundant biconnected components above the root, or vertices that never occur as
LCAs of leaf pairs may leave $\rel_G$ unchanged while changing the graph; see
Figure~\ref{fig:noEncode} for illustrative examples. Hence, if one wants to recover a graph from its
LCA-constraints, one has to restrict to suitable classes of DAGs or networks.

\begin{figure}[t]
  \centering
  \includegraphics[width=0.8\textwidth]{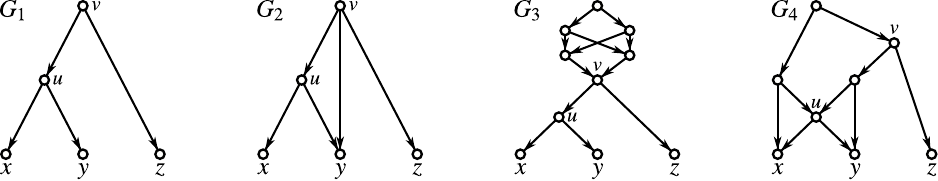}
  \caption{Shown are four non-isomorphic DAGs $G_1$, $G_2$, $G_3$ and $G_4$ with leaf set
  $X=\{x,y,z\}$ for which $\rel_{G_1} = \rel_{G_2} = \rel_{G_3}=\rel_{G_4}$. 
  For each $i\in\{1,2,3,4\}$, it holds that 
  $u=\lca_{G_i}(xy)$ and $v=\lca_{G_i}(xz)=\lca_{G_i}(yz)$. 
  Moreover, all vertices in $V(G_i)\setminus (X\cup \{u,v\})$, if there are any, 
  do not serve as LCA for any $a,b\in X$.
  The DAGs $G_1$ and $G_2$ are 2-lca-relevant, but $G_3$ and $G_4$ are not. All DAGs have the 2-lca-property. 
  Moreover, $G_1$ and $G_3$ are shortcut-free, but $G_2$ and $G_4$ are not.
  }
  \label{fig:noEncode}
\end{figure}

The first main idea of this paper is to separate the information contained in $\rel_G$ from the
parts of $G$ that are invisible to pairwise LCA-comparisons. To this end, we compare two graphs
associated with a DAG $G$. The first is the canonical DAG $\cG_{\rel_G}$ constructed directly from
the relation $\rel_G$. The second is the 2-regularization $\reg(G)$, obtained from $G$ by removing
all vertices that are not LCAs of one or two leaves and then deleting shortcuts. 
In Section~\ref{sec:CanonG} we
show that these
two constructions coincide up to isomorphism, i.e. that $\reg(G)\simeq \cG_{\rel_G}$ holds. Consequently, the
obstruction to being encoded by $\rel$ lies precisely in the information lost when passing from $G$
to $\reg(G)$. This yields a simple reconstruction principle: a class of DAGs is encoded by $\rel$
whenever no information is lost during 2-regularization, that is, $G=\reg(G)$ for all $G$ in the class, or
the information lost in passing from $G$ to $\reg(G)$ can be recovered canonically within the class.

We apply this principle to several natural classes of phylogenetic networks. In particular, in
Section~\ref{sec:encoding-principle}, we show that shortcut-free and 2-lca-relevant DAGs are encoded
by $\rel$. This immediately recovers the tree case and yields encoding results for several types of
level-1 networks. We then also consider cluster-theoretic conditions and prove that regular networks
with binary clustering systems, as well as regular networks whose clustering systems are closed weak
hierarchies, are encoded by their LCA-constraints. In Section~\ref{sec:normal} we focus on the
well-studied class of normal networks \cite{Willson2010,BSS:26,Francis:25}. Although normal networks
are not encoded by $\rel$ in general, we show that strong-phylogenetic (i.e. those with no vertices of out-degree one) 
normal networks $N$ are
encoded since for such networks since $N=\reg(N)$ holds. Moreover, even though $N$ is not isomorphic
to $\reg(N)$ for $N$ in the class of separated phylogenetic normal networks, we show that this class
is still encoded by $\rel$, since we can still recover the information that is lost in
2-regularization. Our results also imply that binary normal networks are encoded by $\rel$.

For a DAG $G$ with leaf set $X$ the relation $\rel_G$ is defined on every pair of elements in $X$.
Thus it is of interest to consider if we can still encode $G$ if we have less information. In
Section~\ref{sec:triple} we study a sparse version of the LCA-relation, namely the relation
$\rel_G^3$ obtained by restricting $\rel_G$ to those $(ab,ac)\in\rel_G$ with $a,b,c\in X$ pairwise
distinct. This relation is close in spirit to displayed rooted triples in a network (see e.g.
\cite{francis:25b}), but records LCA-comparisons rather than topological subtrees that are
displayed. For DAGs with the 2-lca-property, we show that $\rel_G^3$, together with the fixed leaf
set $X$, determines the full relation $\rel_G$. This, in turn, yields polynomial-time methods to
reconstruct DAGs and networks, up to isomorphism, from $\rel_G^3$ within several of the graph
classes considered here.

The rest of this paper is organized as follows. Section~\ref{sec:prelim} recalls the necessary terminology
concerning DAGs, phylogenetic networks, clusters, least common ancestors, shortcuts, regular
networks, and the $\ominus$-operator which is used to define 
2-regularization. In Section~\ref{sec:lca-rel-basics} we 
give the definition and basic properties of LCA-constraints, and then continue in 
in Section~\ref{sec:CanonG} by introducing the canonical DAG $\cG_{\rel_G}$
associated with $\rel_G$ for a DAG $G$ and comparing it with  the 2-regularization $\reg(G)$,
in particular proving that $\cG_{\rel_G}$ and $\reg(G)$ are isomorphic.
We then use this latter 
result in Section~\ref{sec:encoding-principle}
to derive a general encoding principle and show that, amongst other classes, shortcut-free 2-lca-relevant DAGs,
phylogenetic trees, regular level-1 networks, and regular networks defined by binary clustering
systems are encoded by $\rel$. In Section~\ref{sec:normal}
we consider normal networks and show that, although normal
networks are not encoded by $\rel$ in general, several important subclasses are. 
This includes, amongst others, binary normal networks.
In Section~\ref{sec:triple} we consider the sparse triple-like relation $\rel_G^3$ and prove 
encoding and polynomial-time reconstruction results from $\rel_G^3$ for 
network classes satisfying the 2-lca-property. 
We conclude in Section~\ref{sec:outlook} with a summary and some open problems and directions for future research.
In particular, we have summarized the main results
established in this paper in Table~\ref{tab:lca-encoding-summary}.


\section{Basics}
\label{sec:prelim}

\paragraph{Sets, Relations and Clustering Systems.}
In what follows, $X$ will always denote a finite non-empty set.
We let  $\pairs(X)\coloneqq\{\{a,b\} \mid a,b\in X\}$ denote the  set 
system consisting of all 1- and 2-element subsets of $X$.
We will often write $ab$ and $aa$ for the elements $\{a,b\}$ and
$\{a\}$ in $\pairs(X)$, respectively. Thus, $ab=ba$ always holds.

Given a set $A$, a \emph{(binary) relation $R$ (on $A$)} is a subset $R\subseteq A\times A$.
Furthermore, we define the \emph{support } of a relation $R$ on $A$ as $\support_R \coloneqq \{p\in
A \mid \text{ there is some } q\in A \text{ with } (p,q)\in R \text{ or } (q,p)\in R \}$, that is,
the subset of $A$ that contains precisely those $p\in A$ that are in $R$-relation with some $q\in
A$. Note that a \emph{poset} $(A, \sqsubseteq)$ is a set $A$ equipped with a \emph{partial order}
$\sqsubseteq$, i.e., a binary relation $\sqsubseteq$ on $A \times A$ that is reflexive, transitive,
and anti-symmetric.

A \emph{clustering system} $\mathfrak{C}$ on $X$ is a set of subsets of $X$
such that $\{x\}\in \mathfrak{C}$ for all $x\in X$, $\emptyset\notin \mathfrak{C}$ and $X\in \mathfrak{C}$. 
Let $\mathfrak C$ be a clustering system on $X$. Then, 
$\mathfrak{C}$ is \emph{closed}, if 
for all $C,C' \in \mathfrak{C}$ with $C\cap C'\ne\emptyset$ we have
                    $C\cap C'\in \mathfrak{C}$.
Moreover, $\mathfrak C$ is
\emph{binary} if the following two conditions hold:
\begin{enumerate}[noitemsep]
    \item For every $x,y\in X$,  there is a unique inclusion-minimal
    cluster $C\in\mathfrak C$ with $\{x,y\}\subseteq C$.
    \item For every $C\in\mathfrak C$, there exist $x,y\in C$ such that
    $C$ is the unique inclusion-minimal cluster in $\mathfrak C$ containing
    $\{x,y\}$.
\end{enumerate}
Two sets $A$ and $B$ \emph{overlap} if $A\cap B\notin \{\emptyset, A,B\}$; a
clustering system is a \emph{hierarchy} if it does not contain 
overlapping sets. Furthermore, $\mathscr{C}$ is a \emph{weak hierarchy} if $C_1\cap C_2\cap C_3 \in\{C_1\cap
    C_2,C_1\cap C_3, C_2\cap C_3\}$ for all $C_1,C_2,C_3\in\mathscr{C}$. Note that a
    hierarchy is always a weak hierarchy, but not conversely \cite{Bandelt:89}.

\paragraph{DAGs and networks}
A \emph{directed graph $G=(V,E)$} is a tuple with non-empty vertex set $V(G)\coloneqq V$ and arc set
$E(G)\subseteq V\times V$. We let $\outdeg_G(v)\coloneqq\left|\left\{u\in V \colon (v,u)\in
E\right\}\right|$ and $\indeg_G(v)\coloneqq\left|\left\{u\in V \colon (u,v)\in E\right\}\right|$
denote the \emph{out-degree} and \emph{in-degree} of a vertex $v$ in $V$, respectively. A vertex $v$
in $G$ is a \emph{leaf} of $G$ if $\outdeg_G(v)=0$, it is a \emph{root} of $G$ if $\indeg_G(v)=0$,
and it is a \emph{hybrid} if $\indeg_G(v)\geq 2$. Note that leaves might also be roots or hybrids.
For vertices $u,v$ in $G$, we write $v\preceq_G u$ if and only if there is a directed $uv$-path in
$G$. If $v\preceq_G u$ and $v\neq u$, we write $v\prec_G u$. If $u\to v$ is an arc in $G$, then $v$
is a \emph{child} of $u$ and $u$ a \emph{parent} of $v$. Here $\parent_G(v)$ denotes the set of
parents of $v$ in $G$. Moreover, if $u\preceq_G v$, we call $u$ a \emph{descendant} of $v$ and $v$
an \emph{ancestor} of $u$. 
	
Directed graphs without directed cycles are called \emph{directed acyclic graphs (DAGs)}. A
\emph{network} is a DAG with a unique root. A \emph{(rooted) tree} is a network that does not
contain hybrid vertices. A \emph{biconnected component} of a network is an inclusion-maximal
subgraph that remains connected after the removal of any single vertex. In a network $N$, every
biconnected component has a unique $\preceq_N$-maximal vertex \cite[Lem.~8]{Hellmuth2023}. Following
\cite{Hellmuth2023}, a network $N$ is \emph{level-1} if each inclusion-maximal biconnected subgraph
of $N$ contains at most one hybrid vertex distinct from its $\preceq_N$-maximal vertex, which itself
may or may not be hybrid.

If $G$ is a DAG with leaf set $X$, then $G$ is a \emph{DAG on $X$}. Given such a DAG $G$, for leaves
$x,y\in X$, a vertex $v\in V(G)$ that is an ancestor of both $x$ and $y$, is a \emph{common ancestor
of $x$ and $y$}. Moreover, $v$ is a \emph{least common ancestor} (LCA) of $x$ and $y$ if $v$ is a
$\preceq_G$-minimal common ancestor of $x$ and $y$. In case $v$ is the unique least common ancestor
of $x$ and $y$ we write $v = \lca_G(xy)$ and say that \emph{$\lca_G(xy)$ is well-defined};
otherwise, we leave $\lca_G(xy)$ undefined. We now list some further properties of DAGs on $X$ used
throughout this paper: 

\begin{itemize}
\item[] A DAG $G$ on $X$ \dots
\begin{itemize}
	\item[\dots] is \emph{phylogenetic} if it does not contain a vertex $v$ such that $\outdeg_G(v)=1$ and
		$\indeg_G(v)\leq 1$.
    \item[\dots] is \emph{strong-phylogenetic} if it does not contain a vertex $v$ such that $\outdeg_G(v)=1$
	\item[\dots] is \emph{separated} if all hybrids $v$ in $G$ have $\outdeg_G(v)=1$ \cite{Hellmuth2023}. 
	
	\item[\dots] is \emph{binary} if every tree vertex $v$ is
								either a leaf or has $\outdeg_G (v) = 2$, and every hybrid vertex
								$v$ satisfies $\indeg_G (v) = 2$ and $\outdeg_G(v) = 1$.
     
	\item[\dots] is  \emph{2-lca-relevant} if, for all $v\in V(G)$, 
						    there are 
						    $x,y\in X$ such that $v=\lca_G(xy)$ \cite{HL:24}. 
	\item[\dots] has the \emph{2-lca-property} if, for all $x,y\in X$, 
					      the LCA $\lca_G(xy)$ is well-defined \cite{HL:26}.
	\item[\dots] is  \emph{shortcut-free} if it does not contain \emph{shortcuts}, i.e., 
							 arcs $e=(u,w)$  and a directed $uw$-path that does not contain the arc $e$
							\cite{linz2020caterpillars, DOCKER2019129}.
	\item[\dots] is \emph{tree-child} if every vertex $v \in  V(G) \setminus X$ has a 
						  	child $u$ of $v$ with $\indeg_G(u) = 1$  \cite{Cardona:09}. 
	\item[\dots] is \emph{normal} if it is tree-child and shortcut-free \cite{Willson2010}. 			
\end{itemize}
\end{itemize}
    
Throughout the paper, we will use the \emph{shortcut-free version} $G^-$ of a DAG $G$
that is obtained from $G$ by removal of all of its shortcuts.
\begin{lemma}[{\cite[L.~2.5]{HL:24}}]\label{lem:properties-SF-G_NEW}
    Let $G$ be a DAG on $X$. Then, $G^-$ is a shortcut-free DAG on $X$. Moreover, $V(G)=V(G^-)$ and,
    for all $u, v \in V(G)$, we have $u \preceq_{G} v$ if and only if $u \preceq_{G^-} v$. 
\end{lemma}

Two DAGs $G$ and $H $ are \emph{graph isomorphic} if there is a graph isomorphism between $G$ and
$H$, i.e., a bijective map $\varphi\colon V(G)\to V(H)$ such that $(u,v)\in E(G)$ if and only if
$(\varphi(u),\varphi(v))\in E(H)$. In this case, we write $G\approx H$. Moreover, if $G$ and $H$ are
DAGs that also have the same leaf set $X$, then we say that $G$ and $H$ are \emph{isomorphic}, in
symbols $G\simeq H$, if there is a graph isomorphism $\varphi$ between $G$ and $H$ that satisfies
$\varphi(x) = x$ for all $x\in X$.

\paragraph{Hasse diagram and regular networks.}

For a poset $({Q},\sqsubseteq)$, the \emph{Hasse diagram} $\Hasse(Q,\sqsubseteq)$ is the DAG with
vertex set $Q$ and arcs $(A,B)$ if (i) $B\sqsubseteq A$ and $A\neq B$ and (ii) there is no $C\in Q$
with $B\sqsubseteq C\sqsubseteq A$ and $C\neq A,B$. Since the partial order $\sqsubseteq$ is
transitive and anti-symmetric, $\Hasse(Q,\sqsubseteq)$ is always a DAG. Moreover, by definition,
$\Hasse(Q,\sqsubseteq)$ is shortcut-free.

We now relate DAGs to set systems. For a DAG $G$ on $X$ and a vertex $v\in V(G)$ we put
$\CC_G(v)=\{x\in X\mid x\preceq_G v\}$, i.e., the set of leaves that are descendants of $v$. This
set is called the \emph{cluster} of $v$, and we let $\mathfrak{C}(G)=\{\CC_G(v)\mid v\in V(G)\}$
denote the set system on $X$ comprising all clusters of $G$. Note that in case $G$ is a network,
$\mathfrak{C}(N)$ is a clustering system. Conversely, the Hasse diagram
$\Hasse(\mathfrak{C},\subseteq)$ with $\subseteq$ being the usual subset-relation, provides a
natural way to associate a DAG to a set system $\mathfrak{C}$. This, in turn, gives rise to the
following definition. 
\begin{definition}[{\cite{Baroni:05}}]
  \label{def:regular-N}
  A DAG $G=(V,E)$  is \emph{regular} if the map
  $\varphi\colon V\to V(\Hasse(\mathfrak{C}(G)))$ defined by  $v\mapsto \CC_G(v)$ is a graph
 isomorphism between $G$ and $\Hasse(\mathfrak{C}(G))$, i.e., in symbols, $G \approx \Hasse(\mathfrak{C}(G))$.
\end{definition}

Recall that we write $G\simeq H$ for DAGs $G$ and $H$ on the same leaf set $X$ if there is a graph
isomorphism between them that fixes every leaf in $X$. This differs from the graph-isomorphism
relation $\approx$ used in the definition of regular DAGs, which does not require a common leaf set.
In fact, if $G$ is a DAG on $X$, then the leaves of $\Hasse(\mathfrak C(G),\subseteq)$ are the
singleton clusters $\{x\}$ with $x\in X$, rather than the leaves $x$ themselves. Thus, $G$ and
$\Hasse(\mathfrak C(G),\subseteq)$ cannot be compared by a leaf-fixing isomorphism in the sense of
$\simeq$; in particular, $G\not\simeq \Hasse(\mathfrak C(G),\subseteq)$. Nevertheless, equality of
clustering systems does imply a leaf-fixing isomorphism between regular DAGs on the same leaf set.
To be more precise, if $G$ and $H$ are regular DAGs such that $\mathfrak{C}(G)=\mathfrak{C}(H)$,
then the maps $v\mapsto \CC_G(v)$ and $w\mapsto \CC_H(w)$ identify both DAGs with the Hasse diagram
$\Hasse(\mathfrak C(G),\subseteq)=\Hasse(\mathfrak C(H),\subseteq)$. Hence, there is an isomorphism
from $G$ to $H$ that maps each vertex $v\in V(G)$ to the unique vertex $w\in V(H)$ with
$\CC_H(w)=\CC_G(v)$. For every leaf $x\in X$, we have $\CC_G(x)=\{x\}=\CC_H(x)$, and therefore this
isomorphism fixes $x$. Thus $G\simeq H$. We summarize the latter discussion in:

\begin{observation}\label{obs:regular-same-clusters=>iso}
    If $G$ and $H$ are regular DAGs such that $\mathfrak{C}(G)=\mathfrak{C}(H)$, then $G\simeq H$.
\end{observation}

For the convenience of the reader, in Figure~\ref{fig:class-inclusions}
we summarize some of the interrelationships between the different types of networks
that we have introduced so far.

\begin{figure}[h]
\centering
\includegraphics[width=0.8\textwidth]{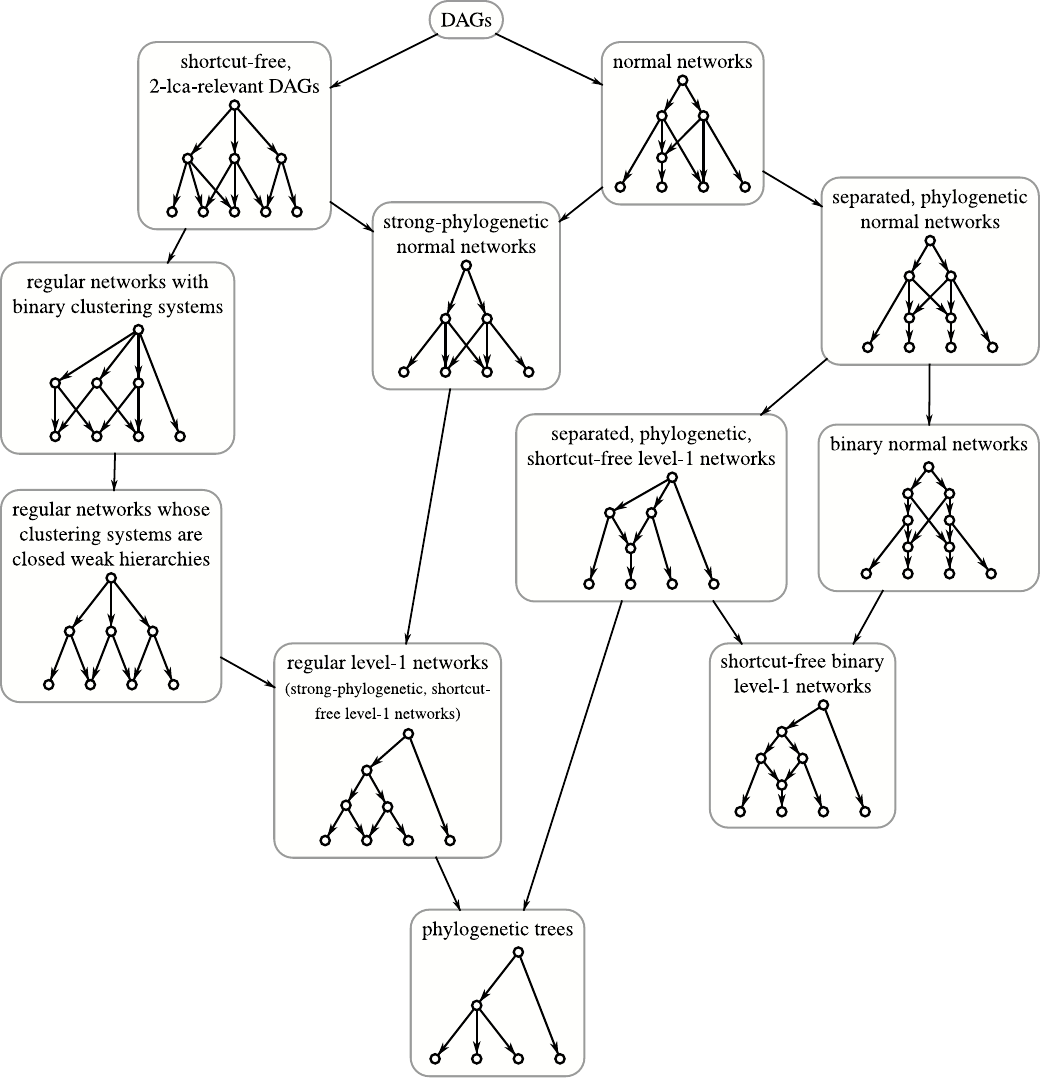}
\caption{Schematic inclusion relations among some of the classes of DAGs and networks
considered in this paper, where solid arrows indicate subclass inclusions.
The inclusion relations follow mainly from Theorem~\ref{thm:classes-poly-recover-from-rel3}. 
The single outgoing arc from regular networks whose clustering systems are closed weak hierarchies is a consequence of 
\cite[Cor~28 \& 29]{Hellmuth2023}. For each class, a representative example network is shown that is not contained in the corresponding subclasses, thereby showing that the indicated inclusions are proper.
}
\label{fig:class-inclusions}
\end{figure}

\paragraph{The $\ominus$-operator.}  
We recall the $\ominus$-operator, introduced and studied in \cite{HL:24,SCHS:24}. Intuitively,
applying $G\ominus v$ removes the vertex $v$ and reconnects each parent of $v$ directly to each
child of $v$

\begin{definition}[{\cite{SCHS:24}}]\label{def:ominus}
  Let $G=(V,E)$ be a DAG and $v\in V$. Then $G\ominus v=(V',E')$ is the directed graph with vertex
 set $V'=V\setminus\{v\}$ and arcs $(p,q)\in E'$ precisely if $v\ne p$, $v\ne q$ and $(p,q)\in E$,
 or if $(p,v)\in E$ and $(v,q)\in E$. 
\end{definition} 

The $\ominus$-operation is order-independent on sets of vertices, i.e., 
for distinct vertices $v,w\in V(G)$,
$   (G\ominus v)\ominus w=(G\ominus w)\ominus v$ (cf.\ \cite[L~2.3]{HLM:26-RegNormArxiv}).
Hence, for a given set non-empty subset $W = \{w_1,\dots,w_\ell\} \subsetneq V(G)$, 
we can, without loss of generality, define 
  \[G\ominus W \coloneqq (\dots ((G \ominus w_1) \ominus w_2) \dots)\ominus w_\ell.\]
For technical reasons, put $G\ominus\emptyset=G$.

The following basic property will be used repeatedly. Applying the
$\ominus$-operation to  DAGs on $X$ preserves the ancestor relation among the vertices that
remain; in particular, the clusters of the remaining vertices are unchanged.
\begin{lemma}[{\cite[Obs~5.3]{HL:24}}]\label{lem:ominus-basics}
    Let $G$ be a DAG on $X$ and $W \subseteq V (G) \setminus X$. 
    Then, $G \ominus W$ is a DAG on $X$ that satisfies  
    $u\preceq_{G} v\iff u\preceq_{G\ominus W} v$ for all $u,v \in V (G \ominus W)=V(G)\setminus W$. In particular, 
    $\CC_G (u) = \CC_{G\ominus W} (u)$ for all $u \in V (G \ominus W)=V(G)\setminus W$.
\end{lemma}

\section{LCA-Relations and Encoding}
\label{sec:lca-rel-basics}

In this section, we begin our investigation of encoding DAGs and networks by LCA relations.
Recall from the introduction
that, for a DAG $G$ with leaf set $X$, we associate a relation
the relation $\rel_G$ on $\pairs(X)$ given by
\[ \rel_G \coloneqq \{(ab,xy) \mid \lca_G(ab),\lca_G(xy) \text{ are well-defined and } \lca_G(ab)\preceq_G \lca_G(xy)\}. \] 
Thus, $\rel_G$ records LCA-comparisons with respect to the ancestor order $\preceq_G$.
A class $\mathcal C$ of DAGs is
\emph{encoded by $\rel$} if, for all $G,H\in\mathcal C$,
\[
    G\simeq H \iff \rel_G=\rel_H.
\]
In what follows, we shall use the following simple observation without
explicit mention.

\begin{observation}\label{obs:subclass-inherits-encoding}
Let $\mathcal C$ be a class of DAGs that is encoded by $\rel$.
Then every subclass $\mathcal C'\subseteq \mathcal C$ is encoded by
$\rel$.
\end{observation}

Note that for encoding we are not required to  restrict to classes on a fixed leaf set. Indeed, for every $x\in X$, we
have $x=\lca_G(xx)$ and hence $(xx,xx)\in\rel_G$. Conversely, every
relation of the form $(xx,xx)$ in $\rel_G$ identifies the leaf $x$.
Therefore, the leaf set can be recovered directly from $\rel_G$ as
$X=\{x\mid (xx,xx)\in\rel_G\}$.  

Moreover, $(ab,xy), (xy,ab)\in \rel_G$
implies, by definition, that $\lca_G(ab)\preceq_G \lca_G(xy)$
and $\lca_G(xy)\preceq_G \lca_G(ab)$ and, therefore, 
$\lca_G(ab)= \lca_G(xy)$. 
Thus, $\rel_G$ also records equality between LCAs. Interestingly,
this equality information can be omitted for encoding questions on a fixed leaf set. Indeed, 
one may equivalently work with the strict relation $\srel_G$ as defined 
in the following result.

\begin{proposition}[{\cite[Prop~46]{LAMSH:25}}]\label{prop:rel=srel}
    For all DAGs $G$ and $H$ on $X$ we have $\rel_G=\rel_H$ if and only if $\srel_G=\srel_H$, 
    where 
    \[\srel_G \coloneqq \{(ab,xy) \mid \lca_G(ab), \lca_G(xy) \text{ are well-defined and } \lca_G(ab)\prec_G \lca_G(xy)\}.\]
\end{proposition}

The assumption that $G$ and $H$ are defined on the same leaf set in Proposition~\ref{prop:rel=srel}
is essential for the strict
relation. Indeed, $\srel_G$ contains no reflexive information $(xx,xx)$ 
and therefore does not, by itself, determine the leaf set. 
 For example, let $G=(\{x\},\emptyset)$ and
$H=(\{x,y\},\emptyset)$. Then
$\rel_G=\{(xx,xx)\}   \neq\{(xx,xx),(yy,yy)\}=\rel_H$,
and $G\not\simeq H$, so the class $\mathcal C=\{G,H\}$ is encoded by
$\rel$. However, $\srel_G=\srel_H=\emptyset$ holds,
so $\srel$ cannot distinguish $G$ and $H$. Thus $\mathcal C$ is not
encoded by $\srel$.

Proposition~\ref{prop:rel=srel} now immediately implies the following.
\begin{corollary}\label{cor:relIFFsrel}
    A class $\mathcal{C}$ of DAGs on the same leaf set $X$
    is encoded by $\rel$ if and only if $\mathcal{C}$ is encoded by $\srel$,
    i.e., for all $G,H\in \mathcal{C}$, we have
    $G\simeq H \iff \srel_G=\srel_H$.
\end{corollary}

In the following, we mainly work with $\rel_G$. By
Corollary~\ref{cor:relIFFsrel}, this entails no loss of generality for
encoding questions on a fixed leaf set. At the same time, 
the leaf set $X$ of a DAG $G$ can be recovered from $\rel_G$
and so we do not need to restrict ourselves to
classes of DAGs or networks with the same leaf set.

\section{The canonical DAG and 2-regularization}
\label{sec:CanonG}

Our general aim is to identify the part of a DAG that is determined by its
LCA-relation $\rel_G$. To this end, we consider two natural constructions.
First, starting from $\rel_G$, one can construct a canonical DAG
$\cG_{\rel_G}$ directly from this relation. Second, one can simplify $G$
itself by removing vertices that are invisible to pairwise LCA-comparisons
and by deleting shortcuts; this gives the 2-regularization $\reg(G)$. The
main result of this section shows that these two constructions coincide up
to isomorphism. As we shall see later, $\reg(G)$ captures precisely the part of 
$G$ that is ``visible'' from $\rel_G$, and separates it from the information 
that is not encoded by pairwise $\rel_G$.

Note that $(ab,xy), (xy,ab) \in \rel_G$ implies $\lca_G(ab)=\lca_G(xy)$.
This ``equality information'' contained in $\rel_G$ allows us to identify leaf pairs with the
same LCA. This is an essential 
observation for defining a canonical DAG directly from $\rel_G$ as follows.

\begin{definition}[{The canonical DAG}]\label{def:qset}
    Let $G$ be a DAG and put $\rel\coloneqq \rel_G$. 
    We define the equivalence relation $\doteqdot$ on $\support_\rel$
    by putting, for all $p,q \in \support_\rel$
    $$p \doteqdot  q \iff (p,q) \in \rel \text{ and } (q,p) \in \rel.$$
    Let $[p]$ denote the equivalence class of $\doteqdot$ that contains $p\in\support_\rel$, 
    and let $Q$ denote the set of all such equivalence classes.
    Define the partial order $\sqsubseteq $ on $Q$ by putting, for all classes
    $[p]$ and $[q]$ in $Q$: $[p] \sqsubseteq  [q] \iff (p,q) \in \rel$.
    
     The \emph{canonical DAG} $\cG_\rel$ is defined 
     as the DAG obtained from the Hasse diagram $\Hasse(Q,\sqsubseteq )$ by relabeling all vertices $[aa]\in Q$ with $a$.        
    \end{definition}

We note that the canonical DAG was introduced formally in
\cite{LAMSH:25} in a more general setting. There, for a relation $R$, the
canonical DAG is defined as $\cG_R$ using the so-called $+$-closure of $R$
as discussed in some more detail in Section~\ref{sec:triple}. Since,
as shown in \cite{LAMSH:25}, $\rel_G$ and its $+$-closure coincide for every DAG $G$,
the DAG $\cG_{\rel}$ considered here is a special case of that construction.
This restricted formulation is sufficient for our
purposes and keeps the presentation simpler. 

We now compare the canonical DAG with a second construction associated
directly with $G$, namely its 2-regularization. This construction removes
the parts of $G$ that are irrelevant for the pairwise LCA-order recorded by
$\rel_G$, that is, vertices that do not occur as unique LCAs of one or two leaves,
and shortcuts.

Given a DAG $G$ on $X$, let \[\lcaV(G)\coloneqq \big\{v\in V(G)\mid  v=\lca_G(xy) \text{ for some } x,y\in X\big\},\]
that is, the set consisting precisely of those vertices of $G$ that occur as
the unique LCA of one or two leaves. In particular, $X\subseteq\lcaV(G)$,
since $x=\lca_G(xx)$ for every $x\in X$. The 2-regularization $\reg(G)$ is
obtained from $G$ by applying the $\ominus$-operator to all vertices not in
$\lcaV(G)$ and then deleting all shortcuts. More specifically:

\begin{definition}[2-regularization {\cite{HLM:26-RegNormArxiv}}]
For a DAG $G$ on $X$, the \emph{2-regularization} of $G$ is
\[\reg(G)\coloneqq \bigl(G\ominus \notlcaV(G)\bigr)^-, \]
where
$\notlcaV(G)\coloneqq V(G)\setminus \lcaV(G)$.
\end{definition}

A simple characterization for when $G$ or $G^-$ coincides with $\reg(G)$ is as follows. 
\begin{lemma}\label{lem:scfree-2lcarel}
For every DAG $G$, the following statements hold:
    \begin{enumerate}
        \item $G$ is 2-lca-relevant if and only if  $G^- = \reg(G)$.
        \item $G$ is 2-lca-relevant and shortcut-free if and only if $G = \reg(G)$.
    \end{enumerate}    
\end{lemma}
\begin{proof}
	  Let $G$ be a DAG. Observe that $G$ is 2-lca-relevant precisely if $\notlcaV(G)=\emptyset$
    which is, if and only if $\reg(G) = (G\ominus \emptyset)^- = G^-$. 
    Hence, Statement (1) holds. 
    
    Assume that $G$ is 2-lca-relevant and shortcut-free. Then 
    $G= G^- = \reg(G)$ by Statement (1). 
    Conversely,  if $G = \reg(G)$ then $G$ must be 
    shortcut-free since $\reg(G)$ is. 
    Hence, $G=G^-$ and Statement (1) implies that $G$ is 
    2-lca-relevant. Thus, Statement (2) holds.    
\end{proof}

 $\reg(G)$ and $\cG_{\rel_G}$ have several structural properties in common,
 as shown next.
\begin{proposition}[{\cite{HL:24,LAMSH:25}}]\label{prop:some-properties}
    Let $G$ be a DAG on $X$. Then, $\reg(G)$ and $\cG_{\rel_G}$ 
    are phylogenetic, 2-lca-relevant and regular DAGs on $X$. 
    Moreover, $\reg(G)$ can be constructed from $G$ in polynomial time in $|V(G)|+|E(G)|$ 
    and $\cG_{\rel_G}$  from $\rel_G$ in polynomial time in $|X|$. 
\end{proposition}
\begin{proof}
    Let $G$ be a DAG on $X$. This proof is based on a collection of several previously established
    results. Theorem~5.5 of \cite{HL:24} establishes that $\reg(G)$ is a 2-lca-relevant (which, in
    the notation of \cite{HL:24}, would be denoted $\{1,2\}$-lca-relevant) DAG on $X$. This in
    particular shows that each vertex of $\reg(G)$ is the unique LCA of some subset $A\subseteq X$
    of arbitrary size, equivalent to the notion of being $\{1,2,\ldots,|X|\}$-lca-relevant or,
    simply, lca-relevant. Since $\reg(G)$, by definition, is shortcut-free, the latter and
    Theorem~4.10 of \cite{HL:24} implies that $\reg(G)$ is regular. Observation~7.6 of \cite{HL:24}
    provides that $\reg(G)$ can be constructed from $G$ in polynomial time in $|V(G)|+|E(G)|$. 

    The canonical DAG $\cG\coloneqq\cG_{\rel_G}$ is a phylogenetic and 2-lca-relevant DAG on $X$
    according to Proposition~28 of \cite{LAMSH:25}, which is applicable since $(ab,xx)\notin\rel_G$
    for all $ab\neq xx$ (a property called \textbf{(X1)} in \cite{LAMSH:25}). Since $\cG$ is, by
    definition, shortcut-free, analogous arguments as for $\reg(G)$ ensures that $\cG$ is regular.
    Theorem~38 of \cite{LAMSH:25} shows that $\cG$ can be constructed from $\rel_G$ in polynomial
    time in $|X|$.
\end{proof}

Proposition~\ref{prop:some-properties} collects structural properties of the two DAGs $\reg(G)$ and
$\cG_{\rel_G}$. We now relate these two constructions more explicitly. In particular, the next
theorem shows that the canonical DAG associated with $\rel_G$ and the 2-regularization of $G$ are
not merely both regular DAGs on the same leaf set, but are in fact isomorphic. This theorem also
characterizes the clusters of $\cG_{\rel_G}$ in terms of the 2-lca-relevant vertices of the original
DAG $G$.

\begin{theorem}\label{thm:cluster-iso-of-canG_NEW}
    Let $G$ be a DAG on $X$ and $\mathfrak{C}^{\lcaV}(G) = \{\CC_G(v)\mid v\in \lcaV(G)\}$. 
    Then,
    \[\mathfrak{C}(\cG_{\rel_G})=\mathfrak{C}^{\lcaV}(G) \quad\text{and}\quad  \reg(G) \simeq \cG_{\rel_G}  \approx \Hasse(\mathfrak{C}^{\lcaV}(G),\subseteq).\]
		In particular, $\reg(G)$ can be reconstructed, up to isomorphism, from
		$\rel_G$ in polynomial time in $|X|$.
\end{theorem}
\begin{proof}
    Let $G$ be a DAG on $X$ and put $\cG\coloneqq\cG_{\rel_G}$ to be the canonical DAG of $\rel_G$.
    Let $(Q,\sqsubseteq )$ be the poset of $\rel_G$ as specified in Definition~\ref{def:qset}. By
    Proposition~\ref{prop:some-properties}, $\cG$ is a DAG on $X$. Hence, by construction of the
    canonical DAG, $\cG$ consist of the leaves in $X$ and the non-leaves $[ab]\in Q$ for which
    $a\neq b$. Clearly, $\CC_{\cG}(x)=\{x\} = \CC_G(x)$ for all $x\in X$, since both $G$ and $\cG$
    are DAGs on $X$. 
    
    Now, let $[ab]$ be some vertex of $\cG$ for which $a,b\in X$ and $a\neq b$. Note that by
    definition of $\rel_G$, $[ab]\in Q$ implies that $\lca_G(ab)$ is well-defined. We will now
    proceed to show that $\CC_G(\lca_G(ab)) = \CC_{\cG}([ab])$. To see that
    $\CC_G(\lca_G(ab))\subseteq\CC_{\cG}([ab])$, let $x\in\CC_G(\lca_G(ab))$. By definition,
    $x\preceq_G\lca_G(ab)$. Since $a\neq b$ and $x\in X$ is a leaf, $x=\lca_G(xx)\prec_G\lca_G(ab)$.
    By definition, $(xx,ab)\in\rel_G$ and $(ab,xx)\notin\rel_G$. Hence, $[xx] \sqsubseteq [ab]$ and
    $[xx]\neq[ab]$. It is now easy to verify that there is a path from $[ab]$ to $[xx]$ in
    $\Hasse(Q,\sqsubseteq)$. Thus after relabeling $[xx]$ by $x$, there is a path from $[ab]$ to $x$
    in $\cG$. In particular, $x\preceq_{\cG}[ab]$, which implies that $x\in \CC_{\cG}([ab])$. To see
    that $\CC_{\cG}([ab])\subseteq\CC_G(\lca_G(ab))$, let $x\in \CC_{\cG}([ab])$. By definition,
    $x\in X$ is a leaf and $x\preceq_{\cG}[ab]$. By construction of $\cG$ and definition of
    $(Q,\sqsubseteq )$, it holds that $[xx] \sqsubseteq [ab]$ and $(xx,ab)\in\rel_G$. Therefore,
    $x=\lca_G(xx)\preceq_G\lca_G(ab)$ and thus, $x\in\CC_G(\lca_G(ab))$. Hence,
    $\CC_{\cG}([ab])\subseteq\CC_G(\lca_G(ab))$. In summary, $\CC_G(\lca_G(ab)) = \CC_{\cG}([ab])$
    holds.

    We proceed with showing that $\mathfrak{C}(\cG)=\{\CC_G(v)\mid v\in \lcaV(G)\}$. By the
    preceding arguments, any $C\in \mathfrak{C}(\cG)$ satisfies $C=\CC_{G}(\lca_G(ab))$ for some
    $a,b\in X$. On the one hand, this means that every $C\in\mathfrak{C}(\cG)$ satisfies
    $C=\CC_G(v)$ for some $v\in \lcaV(G)$, i.e., that $\mathfrak{C}(\cG)\subseteq\{\CC_G(v)\mid v\in
    \lcaV(G)\}$. On the other hand, for every vertex $v\in \lcaV(G)$ there exists, by definition,
    either some $a,b\in X$ with $a\neq b$ for which $v=\lca_G(ab)$, or $v=x$ for some $x\in X$. In
    the first case, $[ab]$ is a vertex of $\cG$ and $\CC_G(v)=\CC_G(\lca_G(ab))=\CC_{\cG}([ab])\in
    \mathfrak{C}(\cG)$. In the second case, $\CC_G(v)=xx = \{x\}= \CC_{\cG}(x)\in \mathfrak{C}(\cG)$
    also holds. Thus, $\mathfrak{C}(\cG)=\{\CC_G(v)\mid v\in \lcaV(G)\}$.
 
    It remains to compare $\cG$ with $\reg(G)$. By Proposition~\ref{prop:some-properties}, both
    $\reg(G)$ and $\cG$ are regular DAGs. We show first that
    $\mathfrak{C}(\reg(G))=\mathfrak{C}(\cG)$ to conclude that $\reg(G)\simeq\cG$. To this end, note
    that by definition, we have $V(\reg(G))=V(G\ominus \notlcaV(G))=\lcaV(G)$ and thus
    $\mathfrak{C}(G\ominus \notlcaV(G))=\{\CC_{G\ominus \notlcaV(G)}(v)\mid v\in \lcaV(G)\}$. This
    together with $\mathfrak{C}(\cG)=\{\CC_G(v)\mid v\in \lcaV(G)\}$ and
    Lemma~\ref{lem:ominus-basics} implies that $\mathfrak{C}(G\ominus
    \notlcaV(G))=\mathfrak{C}(\cG)$. Moreover, by Lemma~\ref{lem:properties-SF-G_NEW}
    and since $\reg(G) = (G\ominus \notlcaV(G))^-$,
    $u
    \preceq_{G\ominus \notlcaV(G)} v$ if and only if $u \preceq_{\reg(G)} v$ and, therefore,
    $\mathfrak{C}(\reg(G))=\mathfrak{C}(G\ominus \notlcaV(G))$. Hence, we obtain
    $\mathfrak{C}(\reg(G))=\mathfrak{C}(\cG)$ which, together with
    Observation~\ref{obs:regular-same-clusters=>iso} and the fact that these two DAG are regular,
    implies that $\reg(G)\simeq\cG$. Finally, since $\cG$ is regular and
    $\mathfrak{C}(\cG)=\mathfrak{C}^{\lcaV}(G)$ it follows that $\cG \approx
    \Hasse(\mathfrak{C}^{\lcaV}(G),\subseteq)$.
    
    To see that the last statement in the theorem holds, note that since $\reg(G) \simeq \cG$ and
    since, by Proposition~\ref{prop:some-properties}, $\cG=\cG_{\rel_G}$ can be constructed from
    $\rel_G$ in polynomial time in $|X|$, it follows that $\reg(G)$ can be reconstructed, up to
    isomorphism, from $\rel_G$ by constructing $\cG_{\rel_G}$ in polynomial time in $|X|$.
\end{proof}

As stated Proposition~\ref{prop:some-properties},
$\reg(G)$ can be constructed directly from $G$ in polynomial time. However, the significance of
Theorem~\ref{thm:cluster-iso-of-canG_NEW} is that the structure of $G$ is not required: the relation
$\rel_G$ alone suffices to reconstruct $\reg(G)$, up to isomorphism, in polynomial time in $|X|$.
In particular,
since  $\reg(G) \simeq \cG_{\rel_G}  \approx \Hasse(\mathfrak{C}^{\lcaV}(G),\subseteq)$, 
the 2-regularization $\reg(G)$ and, equivalently the canonical DAG $\cG_{\rel_G}$,
is determined, up to isomorphism, by the
clusters in $\mathfrak C^{\lcaV}(G)$. Hence, whenever a class of DAGs or
networks admits a characterization in terms of its clusters, the preceding
theorem can be used to recognize when $\reg(G)$ belongs to that class. We
now give examples of two such results.

\begin{corollary}\label{cor:canon-trees}
    For every DAG $G$ the following statements are equivalent.
    \begin{enumerate}
        \item $\reg(G)$ is a tree.
        \item $\cG_{\rel_G}$ is a tree.        
        \item $\mathfrak{C}^{\lcaV}(G)$ is a hierarchy. 
    \end{enumerate}
\end{corollary}
\begin{proof}
By Theorem~\ref{thm:cluster-iso-of-canG_NEW}, $\cG_{\rel_G} \simeq \reg(G)$ and 
it immediately follows that Statements (1) and (2) are equivalent. 
We show now that Statements (2) and (3) are equivalent. 
Suppose that $\cG_{\rel_G}$ is a tree. By Proposition~\ref{prop:some-properties}
$\cG_{\rel_G}$ is phylogenetic. By \cite[Thm~3.5.2]{sem-ste-03a}, 
$\mathfrak{C}(\cG_{\rel_G})$ is a hierarchy. 
Theorem~\ref{thm:cluster-iso-of-canG_NEW}
implies that 
$\mathfrak{C}^{\lcaV}(G)$ is a hierarchy.  
Suppose now that $\mathfrak{C}^{\lcaV}(G)$ is a hierarchy. 
By Theorem~\ref{thm:cluster-iso-of-canG_NEW}
$\mathfrak{C}(\cG_{\rel_G})$ is a hierarchy. 
By Proposition~\ref{prop:some-properties}, 
$\cG_{\rel_G}$ is regular, which together
with \cite[Cor~9]{Hellmuth2023} implies that $\cG_{\rel_G}$ is a tree.
\end{proof}

\begin{corollary}\label{cor:canon-lvl1}
    For every DAG $G$ the following statements are equivalent.
    \begin{enumerate}
         \item $\reg(G)$ is a phylogenetic  level-1 network.
        \item $\cG_{\rel_G}$ is a phylogenetic level-1 network.
       \item $\mathfrak{C}^{\lcaV}(G)$ is  closed and satisfies property
                \begin{description}
                    \item[\textnormal{(L):}]  For all $C_1,C_2,C_3\in\mathfrak C^{\lcaV}(G)$,
                if $C_1$ overlaps both $C_2$ and $C_3$, then     
    $C_1\cap C_2 = C_1\cap C_3$.
                \end{description}
    \end{enumerate}
\end{corollary}
\begin{proof}
By Theorem~\ref{thm:cluster-iso-of-canG_NEW}, $\cG_{\rel_G} \simeq \reg(G)$ and 
it immediately follows that Statements (1) and (2) are equivalent. 
We show now that Statements (2) and (3) are equivalent. 
Suppose that $\cG_{\rel_G}$ is a phylogenetic level-1 network. 
By \cite[Thm~8]{Hellmuth2023}, 
$\mathfrak{C}(\cG_{\rel_G})$ is closed and satisfies (L). 
Suppose now that $\mathfrak{C}^{\lcaV}(G)$ is 
is closed and satisfies (L). 
By Theorem~\ref{thm:cluster-iso-of-canG_NEW}
$\mathfrak{C}(\cG_{\rel_G}) = \mathfrak{C}^{\lcaV}(G)$ 
is closed and satisfies (L). 
By Proposition~\ref{prop:some-properties}, 
$\cG_{\rel_G}$ is regular, which together
with \cite[Prop~18]{Hellmuth2023}
implies that $\cG_{\rel_G}$ is a phylogenetic level-1 network
\end{proof}

We conclude this section by noting that 
further results along the lines of Corollary~\ref{cor:canon-trees}
and \ref{cor:canon-lvl1} can be established by referring to \cite{Hellmuth2023}
in which several types of networks are characterized in terms of their clustering  systems. 

\section{A general encoding principle}
\label{sec:encoding-principle}

We now use the connection between $\cG_{\rel_G}$ and $\reg(G)$ given
in the last section to derive a
general criterion for when a class of DAGs is encoded by $\rel$ (see Corollary~\ref{cor:encoding-via-prune}). 
As we shall see in this section, this principle immediately yields 
various classes of well-known networks that are encoded by their LCA relation.

First, observe that neither the canonical DAG nor the 2-regularization alone is,
in general, sufficient to recover the original DAG. Indeed,
Figure~\ref{fig:canonG-relG} shows that
$\cG_{\rel_G}\simeq \cG_{\rel_H}$ does not necessarily imply
$G\simeq H$ or $\rel_G=\rel_H$. Since
$\cG_{\rel_G}\simeq\reg(G)$ for every DAG $G$, the same example also shows
that $\reg(G)\simeq\reg(H)$ need not imply $G\simeq H$ or
$\rel_G=\rel_H$. Nevertheless, equality of the LCA-relations is stronger: 
$\rel_G=\rel_H$ implies $\reg(G)\simeq\reg(H)$, as shown next. 

\begin{figure}[t]
  \centering
  \includegraphics[width=0.8\textwidth]{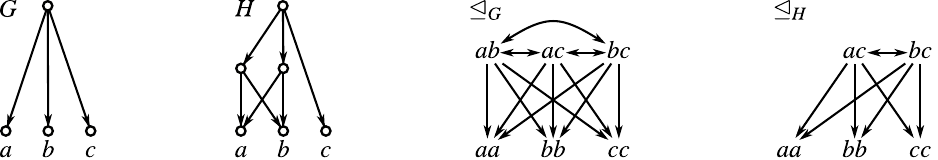}
	 \caption{Two non-isomorphic DAGs $G$ and $H$ with
  $\rel_G\neq \rel_H$ but $\cG_{\rel_G}\simeq \cG_{\rel_H} \simeq G$.
  Here the relations $\rel_G$ and $\rel_H$ are shown graphically, where $(p,q)\in\rel_G$ (resp. $(p,q)\in\rel_H$) is indicated by the arc $q\to p$, although loops $p\to p$ are omitted.
  Note that, $G=\reg(G)\simeq\reg(H)\simeq\cG_{\rel_G}\simeq\cG_{\rel_H}$ although $G\not\simeq H$.}
  \label{fig:canonG-relG}
\end{figure}

\begin{theorem}\label{thm:GsimH=>rel=rel=>ominus-sim_NEW}
Let $G$ and $H$ be DAGs. If $G\simeq H$, then $\rel_G=\rel_H$. Moreover,
if $\rel_G=\rel_H$, then $\reg(G)\simeq \reg(H)$.
\end{theorem}
\begin{proof}
Let $G$ and $H$ be DAGs. Suppose first that $G\simeq H$. Then $G$ and $H$
are DAGs on a common leaf set $X$, and there is a graph isomorphism
$\varphi\colon V(G)\to V(H)$ with $\varphi(x)=x$ for all $x\in X$.

We first note that $\varphi$ preserves the ancestor relation, i.e., for all $u,v\in V(G)$, it holds
that $ u\preceq_G v \iff \varphi(u)\preceq_H \varphi(v)$. Indeed, directed
paths in $G$ are mapped by $\varphi$ to directed paths in $H$, and conversely by $\varphi^{-1}$.
Consequently, for all $x,y\in X$, the LCA $\lca_G(xy)$ is well-defined if
and only if $\lca_H(\varphi(x)\varphi(y))=\lca_H(xy)$ is well-defined. 
Now, if $\lca_G(xy)$ is well-defined, then 
$\varphi(\lca_G(xy))$ is a common ancestor of $x$ and $y$ in $H$.
Moreover, if $u$ is any common ancestor of $x$ and $y$ in $H$, then
$\varphi^{-1}(u)$ is a common ancestor of $x$ and $y$ in $G$.
Thus,  $\lca_G(xy)\preceq_G \varphi^{-1}(u)$ and it follows
that $\varphi(\lca_G(xy))\preceq_H u$. In other words, 
$\varphi(\lca_G(xy))$ is the unique LCA of $x$ and $y$ in $H$.
By analogous arguments, $\varphi^{-1}(\lca_H(xy))$ is
the unique LCA of $x$ and $y$ in $G$.
Hence, $   \varphi(\lca_G(xy))=\lca_H(xy)$ holds. 

Now let $a,b,x,y\in X$. By definition, $(ab,xy)\in\rel_G$ if and only if $\lca_G(ab)$ and
$\lca_G(xy)$ are well-defined and satisfy $ \lca_G(ab)\preceq_G \lca_G(xy)$. By the previous
paragraph and preservation of the ancestor relation, this is equivalent to $
\lca_H(ab)=\varphi(\lca_G(ab)) \preceq_H \varphi(\lca_G(xy))=\lca_H(xy), $ with both LCAs
well-defined. This is precisely $(ab,xy)\in\rel_H$. Hence $\rel_G=\rel_H$.

Now assume that $\rel_G=\rel_H$. By the uniqueness of the canonical DAG constructed from a relation,
we have $ \cG_{\rel_G}=\cG_{\rel_H}. $ Together with Theorem~\ref{thm:cluster-iso-of-canG_NEW}, this
implies $ \reg(G)\simeq \cG_{\rel_G} = \cG_{\rel_H} \simeq \reg(H)$.
\end{proof}

The obstruction to being encoded by $\rel$ is the existence of two
non-isomorphic DAGs $G,H\in\mathcal C$ with $\rel_G=\rel_H$
and thus, by Theorem~\ref{thm:GsimH=>rel=rel=>ominus-sim_NEW}, 
with $ \reg(G)\simeq \reg(H)$.
This yields the following simple sufficient condition for a class of DAGs
to be encoded by its $\rel$-relation.

\begin{corollary}\label{cor:encoding-via-prune}
Let $\mathcal C$ be a class of DAGs such that, for all
$G,H\in\mathcal C$ it holds that
$\reg(G)\simeq \reg(H) $ implies $ G\simeq H.$
Then, $\mathcal C$ is encoded by $\rel$.
\end{corollary}
\begin{proof}
Let $\mathcal C$ be as stated and $G,H\in \mathcal C$. 
If $G\simeq H$, then $\rel_G=\rel_H$ by
Theorem~\ref{thm:GsimH=>rel=rel=>ominus-sim_NEW}. Conversely, assume that
$\rel_G=\rel_H$. Then Theorem~\ref{thm:GsimH=>rel=rel=>ominus-sim_NEW}
implies $\reg(G)\simeq \reg(H)$. By definition of $\mathcal C$, it
follows that $G\simeq H$.
\end{proof}

Note that if $G \simeq \reg(G)$ holds for every DAG $G$ in some class $\mathcal C$,
i.e. no information is lost under 2-regularization, then
for any $G,H\in\mathcal C$ with $\reg(G)\simeq\reg(H)$, we have $G\simeq H$.
Hence Corollary~\ref{cor:encoding-via-prune} immediately implies that 
$\mathcal C$ is encoded by $\rel$. We summarise this in the following:

\begin{corollary}\label{cor:encoding-reg2-fixes}
Let $\mathcal C$ be a class of DAGs such that, for all
$G \in \mathcal C$ we have $\reg(G)\simeq G$.
Then, $\mathcal C$ is encoded by $\rel$.
\end{corollary}

We now exploit this last result to
show that some well-known classes of DAGs and network 
are encoded by their LCA relation.
By Lemma~\ref{lem:scfree-2lcarel}, for all 
2-lca-relevant and shortcut-free DAGs $G$
it holds that $G=\reg(G)$. This together with 
Corollary~\ref{cor:encoding-reg2-fixes} implies:

\begin{theorem}\label{thm:scfree-2lcarel=>encoded_NEW}  
    The class $\mathcal{C}$ of shortcut-free and 2-lca-relevant DAGs is 
    encoded by $\rel$.     
\end{theorem}

We now consider some well-known networks.

\begin{corollary}\label{cor:treeLvl1-relG}
Let $\mathcal{C}$ be either of the following classes:
    \begin{itemize}
        \item The class of phylogenetic trees. 
        \item  The class of regular level-1 networks (or, equivalently, 
        the class of strong-phylogenetic and shortcut-free level-1 networks).
    \end{itemize}    
    Then every $N\in \mathcal{C}$ is 2-lca-relevant and $N=\reg(N)$. 
    In particular,  $\mathcal{C}$ is encoded by $\rel$.
\end{corollary}
\begin{proof}
If $\mathcal C$ is the class of phylogenetic trees, 
then every $N\in \mathcal C$ is shortcut-free and,
by \cite[Cor~7.10]{HL:24}),  2-lca-relevant.

Now suppose that $\mathcal{C}$ is the class of regular level-1 networks. 
By \cite[Prop~15]{Hellmuth2023},  $\mathcal{C}$ 
coincides with the class $\mathcal{C}'$
 of shortcut-free level-1 networks
such that no network in $\mathcal{C}'$ contains vertices with 
out-degree one,  which is precisely the class of 
strong-phylogenetic and shortcut-free level-1 networks.
This together with \cite[Prop~20]{Hellmuth2023} implies that 
any $N$ in $\mathcal{C}$ is 2-lca-relevant.

Theorem~\ref{thm:scfree-2lcarel=>encoded_NEW} 
and Lemma~\ref{lem:scfree-2lcarel} now imply that 
$\mathcal{C}$ is encoded by $\rel$ and $N=\reg(N)$ for all $N\in \mathcal C$,  
for either choice of $\mathcal C$.
\end{proof}

We next apply our encoding approach to networks defined by certain cluster-theoretic conditions. 
For regular networks, binary clustering systems translate directly into
LCA information: they force all pairwise LCAs to be well-defined and ensure
that every vertex occurs as the LCA of one or two leaves.
The same conclusion applies to regular
networks whose clustering systems are closed weak hierarchies, as shown next.

\begin{theorem} \label{thm:regular-binary-is-2lca}
Let $N$ be a regular network such that 
$\mathfrak C_N$ is binary or such that $\mathfrak C_N$ is a closed weak hierarchy. Then, 
$N$ is 2-lca-relevant and has the 2-lca-property 
\end{theorem}
\begin{proof}
    If $N$ is a regular network such that $\mathfrak C_N$ is binary, then 
    \cite[Thm~3.7]{HL:26} implies that $N$ is 2-lca-relevant and has the 2-lca-property. 
    If $N$ is a regular network such that
    $\mathfrak C_N$ is a closed weak hierarchy, then 
    \cite[Prop.~3.1]{Barthelemy:2008} implies that $\mathfrak C_N$ is
    a binary clustering system.  
    Now the assertion follows from the arguments used for regular networks with binary clustering systems. 
\end{proof}

Combining this result with the fact that regular networks are shortcut-free gives
the following  encoding result.

\begin{corollary}\label{cor:regular-binary-encoded}
Let $\mathcal{C}$ be either of the following classes:
    \begin{itemize}
        \item The class of regular networks with binary clustering systems. 
        \item The class of regular networks whose clustering systems are closed weak hierarchies.
    \end{itemize}    
Then for all $N\in \mathcal C$ we have $N=\reg(N)$
and $\mathcal{C}$ is encoded by $\rel$.
\end{corollary}
\begin{proof}
Let $\mathcal{C}$ be the class of regular networks with binary clustering systems 
or the the class of regular networks whose clustering systems are closed weak hierarchies.
By Theorem~\ref{thm:regular-binary-is-2lca}, 
every network in
$\mathcal C$ is 2-lca-relevant and has the 2-lca-property. Moreover, regular networks are
shortcut-free. Lemma~\ref{lem:scfree-2lcarel} implies that $N=\reg(N)$ for all
$N\in \mathcal{C}$ which together with Corollary~\ref{cor:encoding-reg2-fixes}
implies that $\mathcal{C}$ is encoded by $\rel$.
\end{proof}

We conclude this section by noting that, by
Proposition~\ref{prop:some-properties}, $\cG_{\rel_G}$ can be constructed
in polynomial time from $\rel_G$. Moreover, Theorem~\ref{thm:cluster-iso-of-canG_NEW}
implies that $\reg(G) \simeq \cG_{\rel_G}$ for all DAGs $G$. 
This together with the latter results showing when $G=\reg(G)$ holds immediately implies
the following result.

\begin{theorem}\label{thm:poly-from-rel}
Any DAG, resp., network $G$ within the classes considered in
Theorem~\ref{thm:scfree-2lcarel=>encoded_NEW}, Corollary~\ref{cor:treeLvl1-relG} and
\ref{cor:regular-binary-encoded} can be uniquely, up to isomorphism, reconstructed from $\rel_G$ in
polynomial time. 
\end{theorem}


\section{Normal networks}
\label{sec:normal}

In view of Corollary~\ref{cor:encoding-reg2-fixes}, the obstruction
to a DAG $G$ being encoded by $\rel_G$ can be considered to be 
the information that is lost when
passing from $G$ to $\reg(G)$. Hence, to prove that a class $\mathcal C$ is
encoded by $\rel$, it suffices to show either that no information is lost,
that is, $G=\reg(G)$, or that the lost information
can be recovered from $\reg(G)$ within the class
for all $G\in \mathcal C$. In this section, we consider an
example of the latter case. In particular, we show that
certain subclasses of the  class of \emph{normal} networks
are encoded by $\rel$, which also allows us to show 
that some other subclasses of level-1 networks are also encoded by $\rel$.
Note that it was recently shown
that binary normal networks are encoded 
by so-called caterpillar trees on three or four leaves \cite{linz2020caterpillars},
and that a restricted class of binary normal networks is encoded \cite{francis:25b}
by their triplets.
 Our approach can therefore be viewed as a complementary way of encoding such networks.
In contrast to these earlier results, however, we do not need to restrict normal networks to be binary, although we also derive consequences for this important special case later in this section.

\begin{figure}[t]
  \centering
  \includegraphics[width=0.8\textwidth]{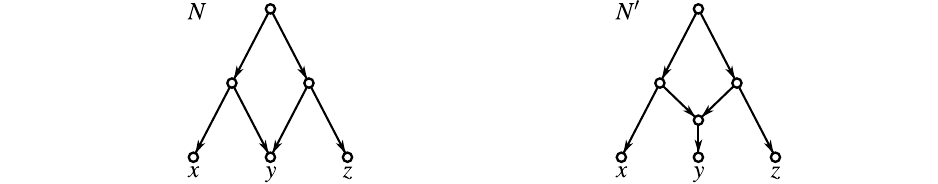}
  \caption{Two non-isomorphic normal networks $N$ and $N'$ on
				  $X=\{x,y,z\}$ with $\rel_N=\rel_{N'}$. Thus, the class of 
                  normal networks is not encoded by $\rel$.
                  The network $N$ is strong-phylogenetic, 2-lca-relevant and normal
                  while $N'$ is separated and normal but not 2-lca-relevant.
                  Both $N$ and $N'$ have the 2-lca-property and it holds that $N=\reg(N)=\reg(N')\neq N'$.
                  }
  \label{fig:normal}
\end{figure}

The class of normal networks  is, in general, not encoded by $\rel$; see
Figure~\ref{fig:normal}. The obstruction is caused by vertices that do not
contribute to any pairwise LCA. In particular, certain vertices of
out-degree one can be removed by the $\ominus$-operator without changing the
underlying LCA-constraints. Thus, in order to obtain positive encoding
results for normal networks, one has to control how such vertices may
occur.
There are two natural ways to do this. One possibility is to exclude
vertices of out-degree one altogether, which leads to the class of strong-phylogenetic
normal networks, also known as \emph{strongly normal} networks \cite{HLM:26-RegNormArxiv}. Another common convention in the phylogenetic
literature is to require hybrid vertices to have a unique child; such
networks are called \emph{separated}. As we shall see, both of these
classes of normal networks are encoded by $\rel$.

We begin with strong-phylogenetic normal networks.

\begin{theorem}\label{thm:str-normal=>iso-to-canonical_NEW}
The class $\mathcal{C}$ of strong-phylogenetic  normal networks 
is encoded by $\rel$ and,  all $N\in \mathcal{C}$ 
are 2-lca-relevant and it holds that $\reg(N)=N$.
Moreover, all $N\in \mathcal{C}$ 
can be uniquely, up to isomorphism, reconstructed from $\rel_N$, 
in polynomial time. 
\end{theorem}
\begin{proof}
Let $N$ be strong-phylogenetic and normal. By \cite[Prop~4.5]{HLM:26-RegNormArxiv},
$N$ is 2-lca-relevant. Since $N$ is normal, it is shortcut-free.
The assertion follows now from Lemma~\ref{lem:scfree-2lcarel}, Theorem~\ref{thm:scfree-2lcarel=>encoded_NEW}
and \ref{thm:poly-from-rel}.
\end{proof} 

The preceding theorem covers normal networks in which no
vertices of out-degree one are present, which
can be conveniently dealt with by observing that for
such networks $N$, $\reg(N)=N$ holds. However, for the 
class of separated normal
networks this is no longer the case; see Figure~\ref{fig:normal}.
However, as we shall now show, the lost information
can be recovered from $\reg(N)$ in case $N$ is a 
separated phylogenetic normal network.

We begin with some useful observations concerning such networks.
First, we note that for phylogenetic normal networks,
the vertices removed by 2-regularization are precisely the hybrid vertices
with a unique child. We shall also use the following facts.

\begin{lemma}[{\cite[L~4.6, Prop~4.7, Cor~4.8, 4.9]{HLM:26-RegNormArxiv}}]\label{lem:ominus-outdeg1-normal}
    Let $N$ be a normal network on $X$. 
    Then, $N\ominus \notlcaV(N)$ is strong-phylogenetic and normal.
    In particular, $\notlcaV(N)=\{v\in V(N)\mid \outdeg_N(v)=1\}$
    and
    $\reg(N) = N\ominus \notlcaV(N)$. If $N$ is, in addition, 
    phylogenetic, then $\notlcaV(N)$ is precisely the set of hybrid vertices in $N$ that have a unique child.
\end{lemma}

We now show that separated phylogenetic normal networks are also encoded
by $\rel$.  
In particular, we show that the information lost by 
2-regularization of such a network can be recovered
by re-inserting one separated hybrid vertex above each
hybrid vertex of the regularized network.

\begin{theorem}\label{thm:sep-normal=>encoded_NEW}
The class of separated and phylogenetic normal networks is encoded by $\rel$. 
\end{theorem}
\begin{proof}
    Let $G$ and $H$ be separated and phylogenetic normal networks on $X$. 
    If $G\simeq H$, then 
    $\rel_G=\rel_H$ follows from Theorem~\ref{thm:GsimH=>rel=rel=>ominus-sim_NEW}. 
    
    Conversely, assume that $\rel_G=\rel_H$. In order to prove that $G\simeq H$, we show first that
    every separated phylogenetic normal networks $N$ is uniquely determined, up to isomorphism, by
    its 2-regularization $\reg(N)$. Let $N$ be a separated phylogenetic normal network on $X$. By
    Lemma~\ref{lem:ominus-outdeg1-normal}, we have $ \notlcaV(N)=\{v\in V(N)\mid \outdeg_N(v)=1\}$
    and that $\reg(N)=N\ominus \notlcaV(N)$  is strong-phylogenetic and normal. Moreover, since $N$ is
    phylogenetic, the set $\notlcaV(N)$ is precisely the set of hybrid vertices of $N$ with a unique
    child. Since $N$ is separated, this is precisely the set of all hybrid vertices of $N$.
		
	We now describe how the deleted vertices in $\notlcaV(N)$ can be recovered from $\reg(N)$. Let
	$u\in\notlcaV(N)$, and let $c$ be the unique child of $u$ in $N$. Since $N$ is tree-child and $u$
	has only the child $c$, the vertex $c$ is a tree-child of $u$. Hence $\indeg_N(c)=1$. When $u$ is
	removed by the operation $N\ominus\notlcaV(N)$, all arcs $p\to u$ in $N$ are replaced by arcs
	$p\to c$. Since $u$ is a hybrid vertex, $|\parent_N(u)|\geq 2$. Consequently, $c$ becomes a hybrid
	vertex in $\reg(N)$. Thus, every vertex $u\in\notlcaV(N)$ gives rise to a hybrid vertex $c$ of
	$\reg(N)$, namely its unique child. We next show that this correspondence is bijective. First, two
	distinct vertices of $\notlcaV(N)$ cannot have the same child since, otherwise, if
	$u_1,u_2\in\notlcaV(N)$ are distinct and have the same child $c$, then $\indeg_N(c)>1$ would yield
	a contradiction. Hence each hybrid vertex of $\reg(N)$ can arise from at most one vertex of
	$\notlcaV(N)$. Conversely, let $c$ be a hybrid vertex of $\reg(N)$. Since $
	V(\reg(N))=V(N)\setminus\notlcaV(N)$, the vertex $c$ belongs to $V(N)$ and is not contained in
	$\notlcaV(N)$. In particular, $c$ is not a hybrid vertex of $N$, because all hybrid vertices of
	$N$ belong to $\notlcaV(N)$. Thus, $\indeg_N(c)\leq 1$. However, since $\indeg_{\reg(N)}(c)> 1$,
	some new parents of $c$ in $\reg(N)$ must have been created by the $\ominus$-operation. This is
	only possible if there exists a vertex $u\in\notlcaV(N)$ that has $c$ as its unique child. By the
	previous arguments, this vertex $u$ is unique. Thus, the vertices of $\notlcaV(N)$ are in
	bijection with the hybrid vertices of $\reg(N)$. Under this bijection, a deleted vertex $u$ is
	represented by its unique child $c$, which is a hybrid vertex of $\reg(N)$. It follows that $N$
	can be reconstructed from $\reg(N)$ by the following canonical operation. For every hybrid vertex
	$c$ of $\reg(N)$, insert a new vertex $u_c$, delete the arcs $ p\to c$ for all
	$p\in\parent_{\reg(N)}(c)$ and add the arcs $p\to u_c$ for all $p\in\parent_{\reg(N)}(c)$ and the
	arc $ u_c\to c $. This operation depends only on $\reg(N)$ and reconstructs $N$ uniquely up to
	isomorphism. Applying this reconstruction to $N=G$ and $N=H$, any isomorphism 
	$\reg(G)\simeq\reg(H)$ maps hybrid vertices of $\reg(G)$ to hybrid vertices of $\reg(H)$ and
	therefore extends uniquely to an isomorphism between the reconstructed networks. Since by
	assumption, $\rel_G=\rel_H$, Theorem~\ref{thm:GsimH=>rel=rel=>ominus-sim_NEW} implies that
	$\reg(G)\simeq \reg(H)$. The latter arguments now imply that $G\simeq H$.
\end{proof}

We conclude this section by considering level-1 networks some of which, as we
now observe, form a subclass of normal networks.

\begin{proposition}\label{prop:phyloSFlvl1=>normal}
Every phylogenetic shortcut-free level-1 network is normal.     
\end{proposition}
\begin{proof}
    By \cite[Prop~9]{Hellmuth2023} every phylogenetic level-1 network is tree-child. 
    Hence, all shortcut-free phylogenetic level-1 networks are normal.     
\end{proof}

\begin{corollary}\label{cor:normal-rel}
The following classes of networks are encoded by $\rel$.
\begin{itemize}
    \item The class of binary normal networks.
    \item The class of separated, phylogenetic and shortcut-free level-1 networks. 
    \item The class of shortcut-free binary level-1 networks. 
\end{itemize}
\end{corollary}
\begin{proof}
By definition, binary networks are phylogenetic and separated.    
Theorem~\ref{thm:sep-normal=>encoded_NEW} thus implies that
the class of binary normal networks is encoded by $\rel$. 

For level-1 networks, observe first that, by Proposition~\ref{prop:phyloSFlvl1=>normal}, 
every shortcut-free phylogenetic level-1 networks is normal. 
Moreover, binary networks are separated and phylogenetic. 
It follows now from Theorem~\ref{thm:sep-normal=>encoded_NEW}
that the class of  separated, phylogenetic and shortcut-free level-1 networks
as well as the class of  shortcut-free binary level-1 networks are encoded by  $\rel$. 
\end{proof}

\begin{remark}
    We could establish here a similar algorithmic result to Theorem~\ref{thm:poly-from-rel} 
    or Theorem~\ref{thm:str-normal=>iso-to-canonical_NEW}
    for the classes considered in Theorem~\ref{thm:sep-normal=>encoded_NEW} and Corollary~\ref{cor:normal-rel}. 
    In fact, all networks $N$ within these classes   can be uniquely, up to isomorphism, reconstructed from $\rel_G$ in  polynomial time.     
    However, we defer this to the next section,
    where we shall prove that these networks are 
    polynomial-time constructable from a subrelation of $\rel$, which will also imply these results.
\end{remark}

\section{Triple-like encodings}
\label{sec:triple}

The preceding encoding results are concerned with the full relation $\rel_G$ of a DAG $G$. 
This relation may contain substantial redundancy: it records all ancestor comparisons among
well-defined pairwise LCAs, including reflexive comparisons such as $(ab,ab)$ or comparisons that can often be
inferred from others. 
It is therefore
natural to ask how much of the information provided by $\rel_G$ is actually needed
to encode a given class of DAGs or networks. By
Corollary~\ref{cor:relIFFsrel}, a class $\mathcal C$ of DAGs on a fixed
leaf set $X$ is encoded by $\rel$ if and only if it is encoded by the
strict subrelation $\srel\subseteq \rel$.

We now consider a different restriction of $\rel_G$, namely the part consisting only of comparisons
of the form $(ab,ac)$, where $a,b,c$ are pairwise distinct leaves. Since such a comparison involves
exactly three leaves, it is natural to view it as analogous to a rooted triple, that is, a binary
phylogenetic tree on three leaves \cite{sem-ste-03a}. However, the information provided by these two objects is
different. A rooted triple specifies which of the three possible binary rooted trees on the leaves
$a,b,c$ is displayed or embedded, whereas a comparison $(ab,ac)\in\rel_G$ only states that $\lca_G(ab)$ is a
descendant of $\lca_G(ac)$.

\begin{definition}
For a DAG $G$ on $X$, we define 
\[
    \rel_G^3
    \coloneqq
    \{(ab,ac)\in\rel_G \colon a,b,c\in X
    \text{ are pairwise distinct}\}.
\]
\end{definition}

We now aim at characterizing classes of DAGs that are encoded by the
$\rel^3$-relation, i.e., those classes $\mathcal{C}$ for which, 
for all 
    $G,H\in \mathcal{C}$, we have 
    \[G\simeq H \iff \rel^3_G=\rel^3_H.\]

\begin{figure}[t]
  \centering
  \includegraphics[width=0.8\textwidth]{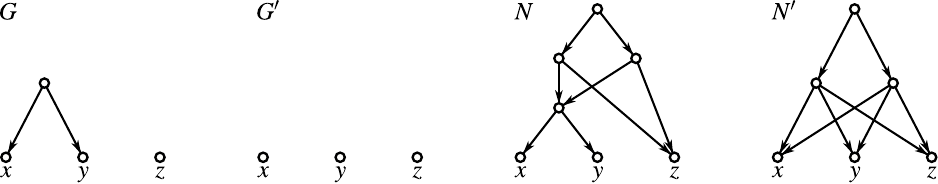}
	\caption{Shown are several DAGs and networks on $X = \{x,y,z\}$.
    For the two DAGs $G$ and $G'$ it holds that $(xx,xy)\in \rel_G$ but 
    $(xx,xy)\notin \rel_{G'}$. Thus, $\rel_G\neq\rel_{G'}$. 
    Since $G\not\simeq G'$ it follows
    that the class  $\mathcal C=\{G,G'\}$ is encoded by $\rel$. 
    However, $\rel_G^3=\emptyset=\rel_{G'}^3$, and hence $\mathcal C$ is not encoded by $\rel^3$. Both $G$ and $G'$ are 2-lca-relevant, but they do not have the 2-lca-property. 
    For the two networks $N$ and $N'$ it holds that $\lca_N(xy)$ is well-defined 
    while $\lca_{N'}(xy)$ is not. Hence, $\rel_N\neq\rel_{N'}$ and, since $N\not\simeq N'$, 
    the class $\mathcal C'=\{N,N'\}$ is encoded by $\rel$. Again, $\rel_N^3=\emptyset=\rel_{N'}^3$
    holds and $\mathcal C'$ is not encoded by $\rel^3$.
    Neither $N$ nor $N'$ is 2-lca-relevant and neither network has the 2-lca-property.
    }
  \label{fig:rel3}
\end{figure}

Figure~\ref{fig:rel3} shows that $\rel^3$ is too weak in general, even
for classes that are encoded by the full relation $\rel$. In particular,
2-lca-relevance alone does not suffice to replace $\rel$ by $\rel^3$. The
problem is that $\rel^3$ only records comparisons involving three distinct
leaves and may therefore fail to record whether the LCA of a given pair of
leaves is well-defined at all. This obstruction disappears under the
stronger 2-lca-property, where all pairwise LCAs are required to be
well-defined.

\begin{theorem}\label{thm:3-rel-encodes-under-2lca}
Let $\mathcal C$ be a class of DAGs on $X$ such that every $G\in\mathcal C$
has the 2-lca-property. If $\mathcal C$ is encoded by 
$\rel$, then $\mathcal C$ is encoded by $\rel^3$. 
\end{theorem}
\begin{proof}
Let $\mathcal C$ be a class  that contain only DAGs on $X$ having
the 2-lca-property. Suppose $\mathcal C$ is encoded by 
$\rel$. 
To show that $\mathcal C$ is encoded by $\rel^3$ we must verify that, for all $G,H\in\mathcal C$ 
we have $ G\simeq H \iff
\rel_G^3=\rel_H^3$. Let $G,H\in\mathcal C$. 
If $G\simeq H$, then $\rel_G=\rel_H$ by Theorem~\ref{thm:GsimH=>rel=rel=>ominus-sim_NEW}, and hence
$\rel_G^3=\rel_H^3$.

Conversely, suppose that $\rel_G^3=\rel_H^3$.
We show first that this implies $\rel_G=\rel_H$. Since every DAG in
$\mathcal C$ has the 2-lca-property and is a DAG on $X$,
the LCA $\lca_G(ab)$ is well-defined for all $a,b\in X$, and the same holds for $H$. 
We will use the latter without explicit mention.
We show now that for all distinct
$a,b\in X$, the cluster $C_{ab}\coloneqq \CC_G(\lca_G(ab))$ is determined by $\rel_G^3$. Clearly,
$a,b\in C_{ab}$. Let $x\in X\setminus\{a,b\}$. We claim that 
\[x\in C_{ab} \iff (ax,ab)\in\rel_G^3.\] 

Suppose that $x\in C_{ab}$. Then $x\preceq_G \lca_G(ab)$ must hold. Since $a\preceq_G
\lca_G(ab)$, it follows that $ \lca_G(ax)\preceq_G \lca_G(ab)$.
Thus $(ax,ab)\in\rel_G$. Since $a,b,x$ are pairwise distinct, it follows that 
$(ax,ab) \in \rel_G^3$. Conversely, assume that $(ax,ab) \in \rel_G^3 \subseteq \rel_G$. 
Therefore,  $\lca_G(ax)\preceq_G\lca_G(ab)$. 
Hence, $a,x \preceq \lca_G(ax)\preceq_G\lca_G(ab)$ and it follows that 
$ x\in C_{ab}$. In summary, for all distinct
$a,b\in X$, the cluster $C_{ab}\coloneqq \CC_G(\lca_G(ab))$ is determined by $\rel_G^3$
and, in particular, 
\[ C_{ab} = \{a,b\}\cup \{x\in X\setminus\{a,b\}\mid (ax,ab)\in\rel_G^3\}. \]
Moreover, it always holds that $C_{aa}\coloneqq \CC_G(\lca_G(aa)) = \{a\}$ for all $a\in X$. 
The same arguments hold for $H$. Since $G$ and $H$ are DAGs on $X$
and since $\rel_G^3=\rel_H^3$, the latter arguments imply that 
$C_{ab} =\CC_G(\lca_G(ab)) = \CC_H(\lca_H(ab))$  for all $a,b\in X$, 
including the case $a=b$.

Now, let $a,b,c,d$ be, not necessarily distinct, elements of $X$.
We show now that
\begin{equation}
  (ab,cd)\in\rel_G \iff    C_{ab} \subseteq C_{cd}. \label{eq:iff1}
\end{equation}
 If $(ab,cd)\in\rel_G$, then $\lca_G(ab)\preceq_G \lca_G(cd)$ and, therefore,
$C_{ab}\subseteq C_{cd}$ holds. Now, assume that $C_{ab}\subseteq C_{cd}$.
We now list some properties of $G'\coloneqq \reg(G)$. 
Since $\lca_G(ab)$ and $\lca_G(cd)$ are well-defined, both vertices belong to 
$\lcaV(G)$ and therefore remain in $G\ominus \notlcaV(G)$. By
Lemma~\ref{lem:ominus-basics} and Lemma~\ref{lem:properties-SF-G_NEW}, passing from $G$ to
$G'$ preserves the ancestor relation among vertices that remain, and preserves all LCAs that
are already well-defined in $G$. In particular, $\CC_{G'}(\lca_{G'}(ab)) = C_{ab}$ and
$\CC_{G'}(\lca_{G'}(cd)) = C_{cd}$ holds. Thus, $
\CC_{G'}(\lca_{G'}(ab)) \subseteq \CC_{G'}(\lca_{G'}(cd))$.
Moreover, $G'$ is regular by Proposition~\ref{prop:some-properties} and, thus, 
the ancestor order in $G'$ is equivalent to inclusion of clusters (c.f. \cite[Prop.~3]{Hellmuth2023}).
Hence, it follows that $ \lca_{G'}(ab) \preceq_{G'} \lca_{G'}(cd)$.
Again using the preservation of the ancestor relation between  $G$ and $G'$ on the 
vertices in  $\lcaV(G)$,  we obtain   $\lca_G(ab)\preceq_G \lca_G(cd)$.
Hence, $(ab,cd)\in\rel_G$.
In summary, Equation~\ref{eq:iff1} holds. 
By the same arguments, we obtain
\begin{equation}
  (ab,cd)\in\rel_H \iff    \CC_H(\lca_H(ab)) \subseteq \CC_H(\lca_H(cd)) \label{eq:iff2}
\end{equation}
for  all $a,b,c,d\in X$. 
As argued above, before Equation~\ref{eq:iff1}, we have $\CC_G(\lca_G(ab)) = \CC_H(\lca_H(ab))$ for all $a,b\in X$. 
This together with Equation~\ref{eq:iff1} and \ref{eq:iff2} implies that 
$(ab,cd)\in\rel_G$ if and only if 
$(ab,cd)\in\rel_H$ for all $a,b,c,d\in X$. 
Thus $\rel_G=\rel_H$.

In summary, we have shown that $\rel_G^3=\rel_H^3$ implies $\rel_G=\rel_H$. 
Since $\mathcal C$ is encoded $\rel$, we
conclude that $G\simeq H$. This completes the proof of the converse direction and therefore the proof.
\end{proof}

\begin{figure}[t]
  \centering
  \includegraphics[width=0.8\textwidth]{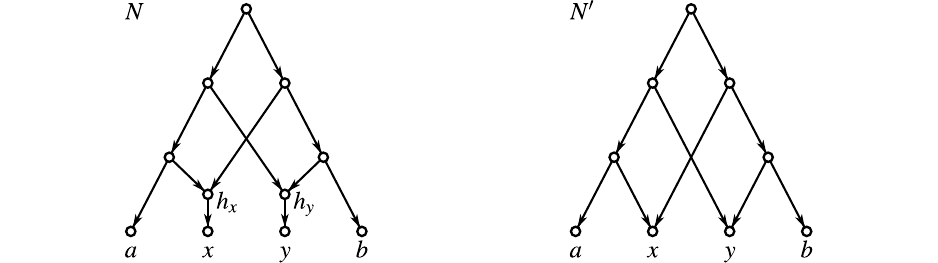}
	\caption{The  network $N$ is binary and normal, and also separated, phylogenetic and normal. 
					The  network $N'$ is strong-phylogenetic and normal. 
						Neither of the networks have the 2-lca-property
					since the LCA of $x$ and $y$ is  not well-defined.
					While $N'$ is 2-lca-relevant, $N$ is not since, for example $h_x \neq \lca_N(cd) $
                    for all $c,d\in X$.
                    }
  \label{fig:binNormal}
\end{figure}

Note that the 2-lca-property does not hold for all of the network
classes that we have considered so far. For example, not every binary normal network
has the 2-lca-property; see Figure~\ref{fig:binNormal}.
For level-1 networks, the situation is different.

\begin{proposition}[{\cite[L~49]{Hellmuth2023}}]\label{prop:lvl1=>2-lca-prop}
    Every level-1 network has the 2-lca-property. 
\end{proposition}

\begin{figure}[t]
  \centering
  \includegraphics[width=0.8\textwidth]{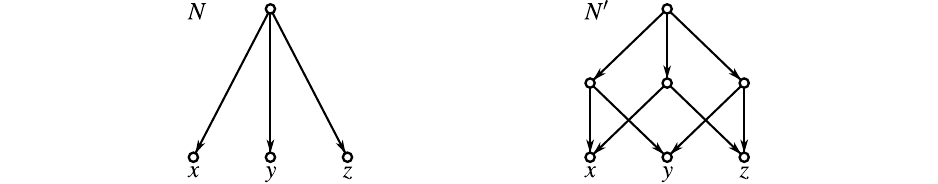}
	\caption{Two non-isomorphic regular networks $N$ and $N'$ 
                that both have the 2-lca-property and 
                for which
					$\rel^3_{N} =\{(ab,ac),(ab,bc),(ac,ab),(ac,bc),(bc,ab),(bc,ac)\}\neq \emptyset=\rel^3_{N'}$ holds. Hence, 
					the class $\mathcal C = \{N,N'\}$ is encoded by
					$\rel^3$. However, $\srel_G^3=\srel_H^3=\emptyset$
					and so $\mathcal C$ is not encoded by $\srel^3$. 
					}
  \label{fig:srel-3}
\end{figure}

In view of Proposition~\ref{prop:rel=srel}, i.e., that 
   for all DAGs $G$ and $H$ on $X$ we have $\rel_G=\rel_H$ if and only if $\srel_G=\srel_H$, 
   it might be tempting to assume that the latter results also generalize to 
 $  \srel_G^3    \coloneqq     \{(ab,ac)\in\srel_G \colon a,b,c\in X    \text{ are pairwise distinct}\}$.
 This however is not the case in general, as illustrated in Figure~\ref{fig:srel-3}.

Now, for a DAG $G$ on $X$ we ask to what extent the sparse relation $\rel_G^3$ determines $G$.
 This requires the closure operation from \cite{LAMSH:25}  that was used implicitly in the general construction of
canonical DAGs. We recall only the rules needed here.

\begin{definition}\label{def:rules}
Let $R \subseteq \pairs(X)\times \pairs(X)$. 
Then, we define
$\support_{R}^+ \coloneqq \support_{R}\ \cup\ \{xx\mid x\in X\}$ and 
the following rules
\begin{description}
  \item[\axiom{R1}] \emph{Reflexivity:} for all  $p\in \support^+_R$, add $(p,p)$ to $R$.
  \item[\axiom{R2}] \emph{Transitivity:} if $(p,q) \in R$ and $(q,r) \in R$, add $(p,r)$ to $R$.
  \item[\axiom{R3}] \emph{Cross-Consistency:} if $ab\in \support_R$ and $(ac,xy) \in R$ and $(bd,xy) \in R$ for some $c,d\in X$, add $(ab,xy)$ to $R$.
\end{description}
\end{definition}

We say that a relation $R$ is, reflexive, resp.,  transitive, resp., cross-consistent
if application of \axiom{R1}, resp., \axiom{R2}, resp., \axiom{R3} does not change $R$. 
The relation $R^+$ denotes the inclusion-minimal relation that contains $R$
and is reflexive, transitive and cross-consistent. As shown in \cite{LAMSH:25}, 
 $R^+$  is uniquely determined and can be computed in polynomial time
 by exhaustive application of \axiom{R1}, \axiom{R2}, and \axiom{R3}. 
 Moreover, the operator $^+ \colon R\, \mapsto\, R^+$ is a closure operator 
 i.e., it satisfies
\emph{extensivity} $R \subseteq R^+$; 
\emph{monotonicity}  $R_1 \subseteq R_2$ $\implies$ $R_1^+ \subseteq R_2^+$;
and \emph{idempotency} $(R^+)^+ = R^+$.

The relation $R=\rel_G^3$ is, in general, not closed under the rules
above, that is, one may have $R\neq R^+$. Nevertheless,  as we show next, 
if $G$ has the
2-lca-property, then $\rel_G^3$ determines the full
relation $\rel_G$ once the leaf set $X$ of $G$ is taken into account. 
In this case, the canonical DAG
$   \cG_{\rel_G}\simeq \reg(G)$
can be reconstructed from $\rel_G^3$. In particular, for classes satisfying
$G=\reg(G)$, this yields a polynomial-time method to recover $G$ uniquely
up to isomorphism from $\rel_G^3$.

\begin{proposition}\label{prop:rel3X-closure-recovers-rel}
Let $G$ be a DAG on $X$ with the 2-lca-property.
Put
\[ \rel_G^{3,X} \coloneqq \rel_G^3 \cup \{(xx,xy),(yy,xy)\mid  \text{ not nec. distinct }  x,y\in X\}. \]
Then
\[ \left(\rel_G^{3,X}\right)^+=\rel_G. \]
\end{proposition}
\begin{proof}
Since $G$ has the 2-lca-property, $\lca_G(xy)$ is well-defined for all $x,y\in X$. 
We will use the latter without explicit mention.
Moreover, $\rel_G^3 \subseteq \rel_G$. Note that in  $G$ it always holds that  
$\lca_G(xx)\preceq_G \lca_G(xy)$ and thus, $(xx,xy)\in \rel_G$ for all $x,y\in X$. 
The latter two arguments imply that $\rel_G^{3,X}\subseteq \rel_G$.
 Since  $\rel_G = \rel_G^+$ (cf.\ \cite[Prop~21]{LAMSH:25}) 
 and the $+$-closure is monotone,  $(\rel_G^{3,X})^+\subseteq \rel_G$.
Also observe that, by definition, $\support_{\rel_G^{3,X}}^+=\pairs(X)$.

We show now that $ \rel_G \subseteq (\rel_G^{3,X})^+$.
 Let $(ab,xy)\in\rel_G$. We consider the two cases of $x,y$ 
 being distinct or not. 
 Suppose first that  $x=y$. Then $(ab,xx)\in\rel_G$ implies
 $\lca_G(ab)\preceq_G \lca_G(xx)=x$. Since $x$ is a leaf, 
 we have $ab=xx$ and thus $(ab,xy)=(xx,xx)$, which belongs to
$\rel_G^{3,X} \subseteq (\rel_G^{3,X})^+$ by definition. 

Assume now that $x\neq y$.
Then $(ab,xy)\in\rel_G$ implies
$   \lca_G(ab)\preceq_G \lca_G(xy)$ and, therefore, $a,b \preceq_G  \lca_G(xy)$.
If $a\in\{x,y\}$, then either $ab=xy$, in which case
$(ab,xy)=(xy,xy)\in(\rel_G^{3,X})^+$ by reflexivity of the $+$-closure, or $a,b,x$ and $y$ involve
three distinct leaves and $(ab,xy)\in (ab,xy)\in \rel_G^{3} \subseteq (\rel_G^{3,X})^+$.

Suppose that $a\neq x$ and $a\neq y$. Then $a \preceq_G  \lca_G(xy)$
and $x \preceq_G  \lca_G(xy)$ implies  $ \lca_G(ax) \preceq_G  \lca_G(xy)$. 
Since, $a,x,y$ are pairwise distinct,  $(ax,xy)\in\rel_G^3$.
Moreover, by definition, $(aa,ax)\in \rel_G^{3,X}$.
Since $(\rel_G^{3,X})^+$ is transitive, the latter arguments imply
 $(aa,xy)\in(\rel_G^{3,X})^+$. 
By similar arguments, $(bb,xy)\in(\rel_G^{3,X})^+$. 
If  $a=b$, then  $(ab,xy)=(aa,xy)\in(\rel_G^{3,X})^+.$
If $a\neq b$, then since $(\rel_G^{3,X})^+ $ is cross-consistent, 
the fact that $(aa,xy), (bb,xy)\in(\rel_G^{3,X})^+$
and $ab\in\pairs(X)=\support_{\rel_G^{3,X}}^+$
implies that  $(ab,xy)\in(\rel_G^{3,X})^+$.

Hence, for all cases, $(ab,xy)\in\rel_G$ implies 
$(ab,xy)\in(\rel_G^{3,X})^+$. In summary, $(\rel_G^{3,X})^+=\rel_G$ holds. 
\end{proof}

Let $G$ be a DAG on $X$ with 2-lca-property. By Proposition~\ref{prop:rel3X-closure-recovers-rel},
we have $(\rel_G^{3,X})^+=\rel_G$. Hence the canonical DAG $\cG_{(\rel_G^{3,X})^+}$ coincides with
$\cG_{\rel_G}$. Combining this with Theorem~\ref{thm:cluster-iso-of-canG_NEW} yields the following.

\begin{proposition}\label{prop:rel3-recovers-reg}
    For all DAGs $G$ on $X$ with 2-lca-property we have
    \[\reg(G) \simeq \cG_{\rel_G} = \cG_{(\rel_G^{3,X})^+}.\]
    In particular, $\reg(G)$
    can be reconstructed, up to isomorphism, from $\rel_G^3$ in polynomial time in $|X|$ by first computing
    $\rel_G = (\rel_G^{3,X})^+$ and, afterwards, computing $\cG_{\rel_G} \simeq \reg(G)$.
\end{proposition}
\begin{proof}
By Proposition~\ref{prop:rel3X-closure-recovers-rel}, we have \((\rel_G^{3,X})^+=\rel_G\). Hence the
canonical DAG constructed from \((\rel_G^{3,X})^+\) is exactly \(\cG_{\rel_G}\).
Theorem~\ref{thm:cluster-iso-of-canG_NEW} then gives $ \reg(G)\simeq
\cG_{\rel_G}=\cG_{(\rel_G^{3,X})^+}. $ The closure \((\rel_G^{3,X})^+\) is computed by exhaustive
application of \axiom{R1}--\axiom{R3}, which is polynomial in \(|X|\), since it is a relation on
\(\pairs(X)\), and \(|\pairs(X)|=O(|X|^2)\), see \cite[Thm~17]{LAMSH:25}. 
Constructing the canonical DAG from the closed relation
is polynomial in $|X|$, by Proposition~\ref{prop:some-properties}.
\end{proof}

If no information is lost under 2-regularization, then the previous results can be combined to provide a
 constructive method to recover a specific DAGs $G$ on $X$ from $\rel^3_G$.

\begin{corollary}\label{cor:sf-2lca+-rel3}
    Let $\mathcal{C}$ be the class of 
    shortcut-free, 2-lca-relevant DAGs on $X$ with the 2-lca-property. 
        Then, for all $G\in \mathcal{C}$ we have
     \[G \simeq \cG_{\rel_G} = \cG_{(\rel_G^{3,X})^+}.\]     
    In particular, every DAG $G\in \mathcal{C}$
    can be reconstructed, up to isomorphism, from $\rel_G^3$ in polynomial time in $|X|$.
\end{corollary}
\begin{proof}
    Let $\mathcal C$ be as stated and $G\in\mathcal C$. 
    Since $G$ is shortcut-free and 2-lca-relevant, 
    Lemma~\ref{lem:scfree-2lcarel} implies that $G=\reg(G)$.
    Moreover, $G$ has the 2-lca-property by assumption. 
    Hence, Proposition~\ref{prop:rel3-recovers-reg} yields
    $G = \reg(G) \simeq \cG_{\rel_G} = \cG_{(\rel_G^{3,X})^+}$
    and thus that $G$ can be reconstructed, up to isomorphism, from $\rel_G^3$ in polynomial time in $|X|$.
\end{proof}

Corollary~\ref{cor:sf-2lca+-rel3} shows that, for shortcut-free and 2-lca-relevant DAGs with the
2-lca-property, the sparse relation $\rel_G^3$ already suffices to reconstruct the entire DAG. We
now apply the same reconstruction principle to the network classes considered above. In the classes
where $G=\reg(G)$ holds, reconstruction follows directly from
Proposition~\ref{prop:rel3-recovers-reg}. In the separated phylogenetic normal case, one first
reconstructs $\reg(G)$ from $\rel_G^3$ and then recovers the deleted vertices by the canonical
reinsertion procedure used in the proof of Theorem~\ref{thm:sep-normal=>encoded_NEW}.

\begin{theorem}\label{thm:classes-poly-recover-from-rel3}
DAGs, respectively networks, with a specified set of leaves $X$ in each of the following classes
are encoded by $\rel^3$ and can be reconstructed, up to isomorphism, from $\rel_G^3$ in
polynomial time in $|X|$. Indented classes are subclasses of the preceding class and are listed
separately only to make explicit that they are covered by this result.
\begin{itemize}
\item Shortcut-free, 2-lca-relevant DAGs with the 2-lca-property.
\begin{itemize}
\item Strong-phylogenetic  normal networks with the 2-lca-property.
\begin{itemize}
\item[--] Regular level-1 networks or, equivalently, strong-phylogenetic and shortcut-free level-1 networks.
\begin{itemize}
    \item[--] Phylogenetic trees.
\end{itemize}
\end{itemize}
\item Regular networks with binary clustering systems.
\begin{itemize}
    \item[--] Regular networks whose clustering systems are closed weak hierarchies.
\end{itemize}
\end{itemize}

\item Separated, phylogenetic normal networks with the 2-lca-property.
\begin{itemize}
    \item Binary normal networks with the 2-lca-property.
    \item Separated, phylogenetic and shortcut-free level-1 networks.
    \begin{itemize}
        \item[--] Shortcut-free binary level-1 networks.
    \end{itemize}
\end{itemize}

\end{itemize}
\end{theorem}
\begin{proof}
We start with showing that the subclass relations listed under the class of shortcut-free,
2-lca-relevant DAGs with the 2-lca-property are correct. Phylogenetic trees are phylogenetic level-1
networks without hybrids. In particular, phylogenetic trees are regular. The class of regular
level-1 networks coincides, by \cite[Prop.~15]{Hellmuth2023}, with the class of strong-phylogenetic
and shortcut-free level-1 networks. Moreover, these networks are 2-lca-relevant by
\cite[Prop.~20]{Hellmuth2023}.  By Proposition~\ref{prop:lvl1=>2-lca-prop}, every level-1 network has
the 2-lca-property. In addition, Proposition~\ref{prop:phyloSFlvl1=>normal} implies that every
shortcut-free phylogenetic level-1 network is normal. Moreover, normal networks are shortcut-free by
definition. Together with Theorem~\ref{thm:str-normal=>iso-to-canonical_NEW}, this implies that
strong-phylogenetic normal networks with the 2-lca-property form a subclass of the class $\mathcal
C$ of shortcut-free, 2-lca-relevant DAGs with the 2-lca-property. 
Moreover,
if $N$ is a regular network such that $\mathfrak C_N$ is a closed weak hierarchy, then $\mathfrak
C_N$ is binary by \cite[Prop.~3.1]{Barthelemy:2008}. Thus regular networks whose clustering systems
are closed weak hierarchies form a subclass of regular networks with binary clustering systems. Finally, by
Theorem~\ref{thm:regular-binary-is-2lca} and definition of regularity, every regular network with a binary clustering system is shortcut-free, 2-lca-relevant and has the 2-lca-property.
Hence all subclasses listed under
$\mathcal C$ are correctly stated and covered by $\mathcal C$. By
Theorem~\ref{thm:scfree-2lcarel=>encoded_NEW}, all DAGs in $\mathcal C$ are encoded by $\rel$. Since
every DAG in $\mathcal C$ has the 2-lca-property, Theorem~\ref{thm:3-rel-encodes-under-2lca} implies
that $\mathcal C$ is encoded by $\rel^3$. Similar to
Observation~\ref{obs:subclass-inherits-encoding}, all listed subclasses of $\mathcal C$ are
therefore also encoded by $\rel^3$. Now, let $G\in\mathcal C$. Since $G$ is shortcut-free and
2-lca-relevant, Lemma~\ref{lem:scfree-2lcarel} implies $G=\reg(G)$. By
Proposition~\ref{prop:rel3-recovers-reg}, the 2-regularization $\reg(G)=G$ can be reconstructed, up
to isomorphism, from $\rel_G^3$ in polynomial time in $|X|$. The same reconstruction applies to all
listed subclasses of $\mathcal C$.

Finally, consider the class $\mathcal C$ of separated and phylogenetic normal networks with the
2-lca-property. First observe that the listed subclass relations are correct. Binary networks are,
by definition, phylogenetic and separated. Hence binary normal networks with the 2-lca-property form
a subclass of $\mathcal C$. Moreover, by Proposition~\ref{prop:phyloSFlvl1=>normal}, every
shortcut-free phylogenetic level-1 network is normal, and by
Proposition~\ref{prop:lvl1=>2-lca-prop}, every level-1 network has the 2-lca-property. Thus
separated, phylogenetic and shortcut-free level-1 networks form a subclass of $\mathcal C$. Finally,
shortcut-free binary level-1 networks form a subclass of separated, phylogenetic and shortcut-free
level-1 networks. By Theorem~\ref{thm:sep-normal=>encoded_NEW}, the class $\mathcal C$ is encoded by
$\rel$. Since all networks in $\mathcal C$ have the 2-lca-property,
Theorem~\ref{thm:3-rel-encodes-under-2lca} implies that $\mathcal C$ is encoded by $\rel^3$. Hence,
similar to Observation~\ref{obs:subclass-inherits-encoding}, all listed subclasses of $\mathcal C$
are encoded by $\rel^3$ as well. It remains to prove the polynomial-time reconstruction for $\mathcal C$
and its subclasses. Let $G\in\mathcal C$. By Proposition~\ref{prop:rel3-recovers-reg}, the
2-regularization $\reg(G)$ can be reconstructed, up to isomorphism, from $\rel_G^3$ in polynomial
time in $|X|$. For separated phylogenetic normal networks, the information removed during
2-regularization can be recovered canonically by the construction used in the proof of
Theorem~\ref{thm:sep-normal=>encoded_NEW}: one re-inserts one separated hybrid vertex above each
hybrid vertex of $\reg(G)$. Therefore, once $\reg(G)$ has been reconstructed from $\rel_G^3$,
applying this canonical reinsertion step reconstructs $G$ up to isomorphism. It remains to note that
this reconstruction is polynomial in $|X|$. Since $\reg(G)$ is, by
Proposition~\ref{prop:some-properties}, 2-lca-relevant, every vertex of $\reg(G)$ is the unique LCA of one
or two leaves. Hence $ |V(\reg(G))|\leq |\pairs(X)|\in O(|X|^2)$. Thus the number of hybrid vertices
of $\reg(G)$ is polynomial in $|X|$, and the canonical reinsertion of one separated hybrid vertex
above each such hybrid vertex can be carried out in polynomial time in $|X|$. Hence every network in
$\mathcal C$, and therefore every network in the listed subclasses of $\mathcal C$, can be
reconstructed from $\rel_G^3$, up to isomorphism, in polynomial time in $|X|$.
\end{proof}

\section{Summary and Outlook}
\label{sec:outlook}

\begin{table}[t]
\centering
\small
\renewcommand{\arraystretch}{1.35}
\setlength{\tabcolsep}{3pt}
\begin{tabular}{p{5.cm}|c|c|c|c|c}
\hline
\textbf{Class $\mathcal C$}
& \textbf{2-lca-rel.}
& \textbf{2-lca-prop.}
& \textbf{$N=\reg(N)$}
& \textbf{$\rel$-enc.}
& \textbf{$\rel^3$-enc.}\\
\hline

\multicolumn{6}{c}{\textit{Level-1 classes, including trees}}\\
\hline

phylogenetic trees
& \cite[Cor~7.10]{HL:24}
& Prop~\ref{prop:lvl1=>2-lca-prop}
& Cor~\ref{cor:treeLvl1-relG}
& Cor~\ref{cor:treeLvl1-relG}
& Thm~\ref{thm:classes-poly-recover-from-rel3}
\\

regular level-1 networks (equ.\ strong-phylogenetic and shortcut-free level-1 networks)
& Cor~\ref{cor:treeLvl1-relG}
& Prop~\ref{prop:lvl1=>2-lca-prop}
& Cor~\ref{cor:treeLvl1-relG}
& Cor~\ref{cor:treeLvl1-relG}
& Thm~\ref{thm:classes-poly-recover-from-rel3}
\\

separated, phylogenetic and shortcut-free level-1 networks
& \large $\times$; \scriptsize Fig.~\ref{fig:normal}
& Prop~\ref{prop:lvl1=>2-lca-prop}
& \large $\times$; \scriptsize Fig.~\ref{fig:normal}
& Cor~\ref{cor:normal-rel}
& Thm~\ref{thm:classes-poly-recover-from-rel3}
\\

shortcut-free binary level-1 networks
& \large $\times$; \scriptsize Fig.~\ref{fig:normal}
& Prop~\ref{prop:lvl1=>2-lca-prop}
& \large $\times$; \scriptsize Fig.~\ref{fig:normal}
& Cor~\ref{cor:normal-rel}
& Thm~\ref{thm:classes-poly-recover-from-rel3}
\\
\hline

\multicolumn{6}{c}{\textit{Regular classes defined by clustering systems}}\\
\hline

regular networks with binary clustering systems
& Thm~\ref{thm:regular-binary-is-2lca}
& Thm~\ref{thm:regular-binary-is-2lca}
& Cor~\ref{cor:regular-binary-encoded}
& Cor~\ref{cor:regular-binary-encoded}
& Thm~\ref{thm:classes-poly-recover-from-rel3}
\\

regular networks whose clustering systems are closed weak hierarchies
& Thm~\ref{thm:regular-binary-is-2lca}
& Thm~\ref{thm:regular-binary-is-2lca}
& Cor~\ref{cor:regular-binary-encoded}
& Cor~\ref{cor:regular-binary-encoded}
& Thm~\ref{thm:classes-poly-recover-from-rel3}
\\
\hline

\multicolumn{6}{c}{\textit{Normal classes}}\\
\hline

strong-phylogenetic normal networks
& Thm~\ref{thm:str-normal=>iso-to-canonical_NEW}
& \large $\times$; \scriptsize Fig.~\ref{fig:binNormal}
& Thm~\ref{thm:str-normal=>iso-to-canonical_NEW}
& Thm~\ref{thm:str-normal=>iso-to-canonical_NEW}
& ?
\\

separated, phylogenetic normal networks
& \large $\times$; \scriptsize Fig.~\ref{fig:normal}
& \large $\times$; \scriptsize Fig.~\ref{fig:binNormal}
& \large $\times$; \scriptsize Fig.~\ref{fig:normal}
& Thm~\ref{thm:sep-normal=>encoded_NEW}
& ?
\\

binary normal networks
& \large $\times$; \scriptsize Fig.~\ref{fig:normal}
& \large $\times$; \scriptsize Fig.~\ref{fig:binNormal}
& \large $\times$; \scriptsize Fig.~\ref{fig:normal}
& Cor~\ref{cor:normal-rel}
& ?
\\

normal networks
& \large $\times$; \scriptsize Fig.~\ref{fig:normal}
& \large $\times$; \scriptsize Fig.~\ref{fig:binNormal}
& \large $\times$; \scriptsize Fig.~\ref{fig:normal}
& \large $\times$; \scriptsize Fig.~\ref{fig:normal}
& \large $\times$; \scriptsize Fig.~\ref{fig:normal}
\\
\hline

\multicolumn{6}{c}{\textit{Classes with 2-lca-property}}\\
\hline

shortcut-free, 2-lca-relevant DAGs with the 2-lca-property.
& by assumption
& by assumption
& L~\ref{lem:scfree-2lcarel}
& Thm~\ref{thm:scfree-2lcarel=>encoded_NEW} 
& Thm~\ref{thm:classes-poly-recover-from-rel3}
\\

strong-phylogenetic  normal networks  with 2-lca-property
& Thm~\ref{thm:str-normal=>iso-to-canonical_NEW}
& by assumption
& Thm~\ref{thm:str-normal=>iso-to-canonical_NEW}
& Thm~\ref{thm:str-normal=>iso-to-canonical_NEW}
& Thm~\ref{thm:classes-poly-recover-from-rel3}
\\

separated, phylogenetic normal networks with the 2-lca-property.
& \large $\times$; \scriptsize Fig.~\ref{fig:normal}
& by assumption
& \large $\times$; \scriptsize Fig.~\ref{fig:normal}
& Thm~\ref{thm:sep-normal=>encoded_NEW}
& Thm~\ref{thm:classes-poly-recover-from-rel3}
\\

binary normal networks with the 2-lca-property
& \large $\times$; \scriptsize Fig.~\ref{fig:normal}
& by assumption
& \large $\times$; \scriptsize Fig.~\ref{fig:normal}
& Cor~\ref{cor:normal-rel}
& Thm~\ref{thm:classes-poly-recover-from-rel3}
\\
\hline

\end{tabular}
\caption{Summary of the results that we obtain in this paper. A reference
in a column indicates where either a property holding for some class or an encoding result is established. 
The symbol $\times$ means that the property or encoding fails in general
and a reference to a figure providing a counterexample is given. The symbol $?$ means that no claim is made here.}
\label{tab:lca-encoding-summary}
\end{table}

We have shown that LCA-constraints provide enough information to encode and reconstruct several
natural classes of DAGs and phylogenetic networks. The central structural result is that the full
relation $\rel_G$ determines precisely the 2-regularization $\reg(G)$, or equivalently the canonical
DAG $\cG_{\rel_G}$. Thus, positive encoding results arise either when $G=\reg(G)$, or when the
information removed by 2-regularization can be recovered canonically within the network class under
consideration. In Table~\ref{tab:lca-encoding-summary} 
we provide a road-map of the results that we have been established here.

This paper opens many questions for future research.
For example, several further network classes remain to be investigated. For instance, one may ask
whether LCA-constraints encode larger subclasses of tree-child networks, networks with closed
clustering systems, semi-regular or cluster networks, or other classes characterized by restrictions
on their clustering systems. 
In particular, the question marks in
Table~\ref{tab:lca-encoding-summary} indicate classes for which we do not currently know whether the
corresponding $\rel^3$-encoding result holds; settling these cases by either positive reconstruction
results or counterexamples would further clarify the scope of LCA-based encodings. Since the
canonical DAG $\cG_{\rel_G}$ is regular and is determined by the clusters of the 2-lca-relevant
vertices, cluster-theoretic characterizations may provide a useful route towards further encoding and
reconstruction results.

More generally, it would be useful to characterize exactly which classes of
DAGs are encoded by $\rel$. Corollary~\ref{cor:encoding-via-prune} gives a sufficient condition in
terms of reconstructability from $\reg(G)$, but it does not provide a full characterization. 
Hence, a natural next step is to identify necessary and sufficient conditions under which a class $\mathcal
C$ is encoded by $\rel$. 

We have shown that, under the 2-lca-property, $\rel_G^3$ determines $\rel_G$ after a
natural closure operation. It remains open to determine whether the 2-lca-property can be weakened,
or replaced by a more intrinsic condition on the support of $\rel_G^3$. In particular, it would be
interesting to characterize those DAGs for which $\rel_G^3$, together with the leaf set, already
determines $\reg(G)$.

For the strict relations $\srel_G$ we observed that $\srel_G$ is equivalent to $\rel_G$ for encoding
questions on a fixed leaf set. However, the analogous statement fails for $\srel_G^3$. This suggests
that the reflexive and equality information contained in $\rel_G^3$ is essential in some cases. A
natural problem is to identify intermediate relations $R_G$ with $ \srel_G^3 \subseteq R_G \subseteq
\rel_G^3 $ that still encode the network classes considered here.

Finally, it would be interesting to compare LCA-constraint encodings more systematically with
other encodings of networks that use displayed subnetworks such as triples, caterpillars and trinets. 
The relation $\rel_G^3$ is triple-like in that it only involves three leaves, but it
records ancestor comparisons of pairwise LCAs rather than displayed topological substructures.
Understanding precisely when these two types of information determine one another could help clarify the
connection between LCA-based and display-based reconstruction methods.

\section*{Acknowledgment}

We thank the organizers of the Bielefeld meeting ``New Directions in Experimental Mathematics'' (June 30 - July
1, 2025), where MH proposed the open problem on LCA constraints for DAGs and networks, and the BIRS Workshop
``Novel Mathematical Paradigm for Phylogenomics'' (August 24-29, 2025), where further discussions helped refine
the ideas presented in this work.

 \section*{Declaration of generative AI in the manuscript preparation process.}

During the preparation of this work, the authors used ChatGPT by OpenAI to assist with the initial
generation of Table~\ref{tab:lca-encoding-summary} and Figure~\ref{fig:class-inclusions}, to help
identify counterexamples to previous conjectures, to support literature searches, and to improve
grammar, spelling, and clarity of the text. After using this tool, the authors reviewed and edited
the content as needed and take full responsibility for the content of the published article.

\bibliographystyle{spbasic}
\bibliography{common_Encode}

\end{document}